\documentclass[9pt, a4paper,twocolumn]{article} 

\usepackage[style=authoryear,backend=biber]{biblatex}
\usepackage{tabularx}
\usepackage{setspace}
\usepackage{subfloat}
\usepackage{mathtools}
\usepackage{amsmath,amsfonts,amssymb,amsthm,mathrsfs, bm}
\usepackage[utf8]{inputenc}
\usepackage[swedish, english]{babel}
\usepackage{csquotes}
\usepackage{subfig}
\usepackage{graphicx}
\usepackage{supertabular}                                             
\usepackage[table]{xcolor} 

\usepackage{graphicx}
\usepackage{epstopdf}
\usepackage{microtype}
\usepackage{multirow}
\usepackage{hyperref}
\usepackage{xcolor}
\usepackage{authblk}
\usepackage{datetime2}	
\usepackage[left=0.5in, right=0.5in, top=0.8in, bottom=0.8in]{geometry} 
\usepackage{setspace}	
\theoremstyle{plain}
\newtheorem{theorem}{Theorem}[section]
\newtheorem{lemma}{Lemma}[section]

\newtheorem{proposition}{Proposition}[section]
\newtheorem{assumption}{Assumption}[section]

\newtheorem{parametrization}{Parametrization}[section]

\usepackage{amsmath,amsfonts,amssymb,amsthm,mathrsfs, bm}
\usepackage{stmaryrd}

\usepackage{graphicx}
\usepackage{tikz}
\usepackage[utf8]{inputenc}
\usepackage{pgfplots} 
\usepackage{pgfgantt}
\usepackage{pdflscape}
\pgfplotsset{compat=newest} 
\pgfplotsset{plot coordinates/math parser=false}

\usepackage{cancel}
\usepackage{endnotes}
\usepackage{commath}
\usepackage{gensymb}
\usepackage{mathtools}
\usepackage{tikz} 
\usetikzlibrary{calc,quotes,angles}
\usepackage{tkz-euclide}
\usetikzlibrary{automata,positioning} 

\usepackage{xcolor}
\usepackage{tikz,tkz-euclide,siunitx}
\usepackage{etoolbox} 

\usetikzlibrary{quotes,arrows.meta,angles,calc,shadings,positioning,babel}

\newtheorem{remark}{Remark}

\usepackage{endnotes}
\usepackage{commath}
\usepackage{gensymb}

\makeatletter
\patchcmd{\tkz@DrawLine}{\begingroup}{\begingroup\makeatletter}{}{}
\makeatother

\allowdisplaybreaks

\DeclareMathOperator{\RE}{Re}

\DeclareMathOperator{\diag}{\mathrm{diag}}

\makeatletter
\newcommand\makebig[2]{%
  \@xp\newcommand\@xp*\csname#1\endcsname{\bBigg@{#2}}%
  \@xp\newcommand\@xp*\csname#1l\endcsname{\@xp\mathopen\csname#1\endcsname}%
  \@xp\newcommand\@xp*\csname#1r\endcsname{\@xp\mathclose\csname#1\endcsname}%
}
\makeatother

\providecommand*{\ped}[1]{%
\ensuremath{_\textnormal{#1}}}

\providecommand*{\eu}%
{\ensuremath{\mathrm{e}}}
\providecommand*{\GammaF}%
{\ensuremath{\mathrm{\Gamma}}}
\providecommand*{\BetaF}%
{\ensuremath{\mathrm{\Beta}}}
\usepackage{bbm}
\makebig{biggg} {3.0}
\makebig{Biggg} {3.5}
\makebig{bigggg}{4.0}
\makebig{Bigggg}{4.5}
\makebig{biggggg}{5.0}
\makebig{Biggggg}{5.5}
\usepackage{longtable}
\usepackage{xcolor}
\DeclareMathSymbol{\Gamma}{\mathalpha}{letters}{"00}
\DeclareMathSymbol{\Delta}{\mathalpha}{letters}{"01}
\DeclareMathSymbol{\Theta}{\mathalpha}{letters}{"02}
\DeclareMathSymbol{\Lambda}{\mathalpha}{letters}{"03}
\DeclareMathSymbol{\Xi}{\mathalpha}{letters}{"04}
\DeclareMathSymbol{\Pi}{\mathalpha}{letters}{"05}
\DeclareMathSymbol{\Sigma}{\mathalpha}{letters}{"06}
\DeclareMathSymbol{\Upsilon}{\mathalpha}{letters}{"07}
\DeclareMathSymbol{\Phi}{\mathalpha}{letters}{"08}
\DeclareMathSymbol{\Psi}{\mathalpha}{letters}{"09}
\DeclareMathSymbol{\Omega}{\mathalpha}{letters}{"0A}

\usepackage{hyperref}
\hypersetup{
    colorlinks = true,
    linkbordercolor = {white},
            linkcolor=blue,
            bookmarks=true,
            filecolor=blue,
            urlcolor=blue,
            citecolor=blue
}
\usepackage{cuted}%
\DeclareMathAlphabet{\altmathcal}{OMS}{cmsy}{m}{n}
\DeclareSymbolFont{cmletters}{OML}{cmm}{m}{it}

\newenvironment{automaticaabstract}
{\par\small\noindent\textbf{Abstract.} }
{\par}

\newcommand{\automatickeywords}[1]{%
\par\small\noindent\textit{Keywords:} #1\par
}

\begin{document}
\title{Stability and stabilization of semilinear single-track vehicle models with distributed tire friction dynamics via singular perturbation analysis}
\date{}
\author[a,b,c]{Luigi Romano\thanks{Corresponding author. Email: luigi.romano@liu.se.}}
\author[b]{Ole Morten Aamo}
\author[c]{Miroslav Krstić}
\author[a]{Jan Aslund}
\author[a]{Erik Frisk}
\affil[a]{\footnotesize{Department of Electrical Engineering, Linköping University, SE-581 83 Linköping, Sweden}}
\affil[b]{\footnotesize{Department of Engineering Cybernetics, Norwegian University of Science and Technology, O. S. Bragstads plass 2, NO-7034, Trondheim, Norway}}
\affil[c]{\footnotesize{Department of Mechanical and Aerospace Engineering, University of California San Diego, La Jolla, CA, 92093, USA}}

\maketitle

\begin{strip}
    \centering
    \begin{minipage}{.8\textwidth}
\begin{automaticaabstract}
This paper investigates the stability and stabilization of semilinear single-track vehicle models with distributed tire friction dynamics, modeled as interconnections of ordinary differential equations (ODEs) and hyperbolic partial differential equations (PDEs). Motivated by the long-standing practice of neglecting transient tire dynamics in vehicle modeling and control, a rigorous justification is provided for such simplifications using singular perturbation theory. A perturbation parameter, defined as the ratio between a characteristic rolling contact length and the vehicle's longitudinal speed, is introduced to formalize the time-scale separation between rigid-body motion and tire dynamics. For sufficiently small values of this parameter, it is demonstrated that standard finite-dimensional techniques can be applied to analyze the local stability of equilibria and to design stabilizing controllers. Whilst the proposed controllers build on classical approaches, the novelty of this work lies in establishing the first singular perturbation framework for ODE-PDE vehicle models with distributed tire dynamics, providing a theoretical justification for their quasi-static reduction and for the use of finite-dimensional tools for analysis and control design.
\end{automaticaabstract}
        \hspace{0.5cm}

\begin{small}
\automatickeywords{Vehicle dynamics; distributed friction models; distributed parameter systems; hyperbolic ODE-PDE systems; semilinear systems; singular perturbation theory}
\end{small}
    \end{minipage}
\end{strip}


\section{Introduction}
\label{sec:introduction}

The analysis and control of road vehicles have long been central topics in both automotive and control engineering. A key aspect of these studies is the modeling of tire–road interaction, which fundamentally determines the stability, handling, and safety of road vehicles. In most of the classical literature, the transient dynamics of tires are neglected, and tires are instead modeled as static nonlinearities that map slip conditions into forces \parencite{Pacejka2,Guiggiani,LibroMio}. This modeling choice has been extremely influential: it simplifies analysis, enables the adoption of low-order vehicle models, and underpins many of the widely adopted techniques for stability analysis, control design, and observer synthesis \parencite{Mojtaba1,Mojtaba3,Shao1,Shao3,Basilio2,Basilio4,LateralControl,Savaresi,Gerdes3,IEEEVT1,IEEEVT2,IEEEVT3}.

As explained in \textcite{Takacs5}, the informal justification behind this simplification lies in a perceived separation of time scales: vehicle rigid-body dynamics evolve relatively slowly compared to the fast rolling contact phenomena taking place inside the tire patch. Building on this consideration, it becomes reasonable to treat tire forces as quasi-stationary functions of slip \parencite{Pacejka2,Guiggiani}, thereby enabling the application of well-established finite-dimensional control methods. The validity of this approximation appears evident if not pleonastic: a vast corpus of literature reports numerical validations, experimental verifications, and successful estimator \parencite{Mojtaba1,Mojtaba3,Shao1,Shao3} and controller implementations \parencite{Basilio2,Basilio4,LateralControl,Savaresi,Gerdes3,IEEEVT1,IEEEVT2,IEEEVT3} based on static tire models. Such models continue to serve as the backbone for advanced driver assistance systems, observer design for lateral dynamics, and even the development of automated driving technologies.

Despite this practical success, the lack of a rigorous theoretical foundation for neglecting tire transients has remained a notable gap. The separation-of-time-scales argument, whilst intuitively appealing, has not been systematically formalized in the context of stability and control of vehicle models. The omission becomes especially relevant when operating conditions depart from those for which static models are most accurate. For instance, recent research by \textcite{Takacs5,Takacs2,Takacs1,Takacs3,Beregi1,Beregi3,BicyclePDE} has highlighted differences between predictions obtained using conventional static tire models and those incorporating transient, distributed tire dynamics. These discrepancies are particularly pronounced at low rolling speeds, where the time scales of tire deformation and rigid-body motion are no longer well separated. Under these conditions, transient tire behavior manifests in nontrivial dynamical effects, including oscillatory instabilities and micro-shimmy phenomena, that static models cannot capture.

To address these limitations, distributed tire models expressed in terms of semilinear hyperbolic partial differential equations (PDEs) have been proposed in \textcite{LibroMio,TsiotrasConf,Tsiotras1,Tsiotras2,Deur0,Deur1,Deur2}. These models capture the spatiotemporal evolution of frictional stresses inside the contact patch and naturally account for the nonstationary generation of longitudinal and lateral forces. Incorporating such descriptions into complete vehicle models yields ODE-PDE interconnections, where the vehicle’s rigid-body states interact with the distributed dynamics of the tire-road interface \parencite{SemilinearV}. Whilst this modeling framework is more realistic and can reproduce subtle dynamical phenomena, it introduces new analytical and control-theoretic challenges: questions of well-posedness, stability, and controller synthesis would involve infinite-dimensional systems \parencite{Weiss,Zwart}, with the established toolbox of finite-dimensional methods being not directly applicable.

Over the past two decades, significant progress has been made in the control of PDE and ODE-PDE interconnections, with powerful techniques developed in areas such as chemical process control, flexible structures, and flow dynamics \parencite{Coron,KrsticBook1,KrsticBook2,OleBook}. Notably, singular perturbation methods, typically restricted to ODEs, have been extended to analyze and stabilize hyperbolic PDE and ODE-PDE systems with multiple time scales in \textcite{Vazquez,Tikhonov1,Tikhonov2,Tikhonov3}. However, the application of such methods to vehicle dynamics has been limited, and a systematic connection between distributed tire models and standard vehicle dynamics theory has been missing.

In this context, the present paper builds upon the distributed models introduced in \textcite{SemilinearV} and on the dissipativity analysis conducted in \textcite{PassExp}. However, neither of these works addressed singular perturbation theory, reduced-order modeling, or the rigorous justification of finite-dimensional vehicle-control methodologies. Indeed, in \textcite{SemilinearV}, the main focus was the derivation and well-posedness analysis of semilinear ODE-PDE vehicle models with distributed tire dynamics, whereas in \textcite{PassExp}, passivity properties were exploited to design fully infinite-dimensional stabilizing controllers. By contrast, the objective of this work is fundamentally different: to establish a singular-perturbation framework connecting distributed-parameter vehicle models with the finite-dimensional tools routinely used in automotive control. This yields rigorous conditions under which classical stability analyses and controller design strategies remain valid despite the presence of distributed friction dynamics. Although some technical arguments inevitably share similarities with those contained in \textcite{SemilinearV, PassExp}, the questions addressed, the methodology adopted, and the resulting theoretical conclusions are substantially different.

Specifically, by leveraging singular perturbation analysis, this paper demonstrates that, for sufficiently small values of a time-scale parameter -- defined as the ratio between a characteristic length scale of the rolling contact process and the vehicle's longitudinal speed -- semilinear ODE-PDE vehicle models can be approximated by reduced finite-dimensional subsystems. In this regime, classical tools of nonlinear stability analysis and controller design can be reliably applied. Importantly, this result formalizes the heuristic justification for neglecting transient tire dynamics at sufficiently high velocities, whilst simultaneously clarifying the conditions under which such simplifications fail.
Building on this foundation, the stabilization problem is also addressed, relying on standard stabilizability assumptions. Although the proposed control strategies themselves are not new (being, in fact, the simplest known), the contribution of this work lies in establishing that their use is theoretically justified for ODE-PDE vehicle models under explicit conditions. In this way, the paper reconciles decades of empirical and simulation-based validation of static tire models with a rigorous mathematical framework.

The contributions of this paper are twofold. First, it provides novel singular-perturbation-based stability results for semilinear vehicle models with distributed tire dynamics, thereby extending the theoretical understanding of vehicle stability into the infinite-dimensional setting. Second, it demonstrates that widely adopted finite-dimensional design methods for vehicle control can be systematically justified when tire dynamics evolve on faster time scales than the rigid-body motion. Albeit derived using simplified ODE-PDE vehicle models, these contributions open new avenues for the analysis and control of ODE-PDE systems with distributed friction.

The remainder of the paper is organized as follows. Section~\ref{sect:Problem} recapitulates the family of semilinear single-track vehicle models with distributed tire friction dynamics considered throughout the manuscript and postulates certain structural assumptions. A stability analysis is then conducted in Section~\ref{sect:stabl} using singular perturbation theory, first by isolating the reduced ODE and boundary-layer PDE subsystems, and then by deriving local stability results for sufficiently small time-scale parameters $\tau_z$. Section~\ref{sect:Controllllll} addresses the stabilization problem. Numerical simulations are discussed in Section~\ref{sect:numerical}. Finally, Section ~\ref{sect:conclusion} concludes the paper with remarks on implications for vehicle dynamics and directions for future research.


\subsection*{Notation}
In this paper, $\mathbb{R}$ denotes the set of real numbers; $\mathbb{R}_{>0}$ and $\mathbb{R}_{\geq 0}$ indicate the set of positive real numbers excluding and including zero, respectively. 
The set of $n\times m$ matrices with values in $\mathbb{F}$ ($\mathbb{F} = \mathbb{R}$, $\mathbb{R}_{>0}$, or $\mathbb{R}_{\geq0}$) is denoted by $\mathbf{M}_{n\times m}(\mathbb{F})$ (abbreviated as $\mathbf{M}_{n}(\mathbb{F})$ whenever $m=n$). $\mathbf{GL}_n(\mathbb{R})$ and $\mathbf{Sym}_n(\mathbb{R})$ represents the groups of invertible and symmetric matrices, respectively, with values in $\mathbb{R}$; the identity matrix on $\mathbb{R}^n$ is indicated with $I_n$. A positive-definite matrix is denoted as $\mathbf{M}_n(\mathbb{R}) \ni Q \succ 0$.
The standard Euclidean norm on $\mathbb{R}^n$ is indicated with $\norm{\cdot}_2$; matrix norms are simply denoted by $\norm{\cdot}$.
$L^2((0,1);\mathbb{R}^n)$ denotes the Hilbert space of square-integrable functions on $(0,1)$ with values in $\mathbb{R}^n$, endowed with inner product $\langle \zeta_1, \zeta_2 \rangle_{L^2((0,1);\mathbb{R}^n)} = \int_0^1 \zeta_1^{\mathrm{T}}(\xi)\zeta_2(\xi) \dif \xi$ and induced norm $\norm{\zeta(\cdot)}_{L^2((0,1);\mathbb{R}^n)}$. The Hilbert space $H^1((0,1);\mathbb{R}^n)$ consists of functions $\zeta\in L^2((0,1);\mathbb{R}^n)$ whose weak derivative also belongs to $L^2((0,1);\mathbb{R}^n)$; it is naturally equipped with norm $\norm{\zeta(\cdot)}_{H^1((0,1);\mathbb{R}^n)}^2 \triangleq \norm{\zeta(\cdot)}_{L^2((0,1);\mathbb{R}^n)}^2 + \norm{\pd{\zeta(\cdot)}{\xi}}_{L^2((0,1);\mathbb{R}^n)}^2$. For a matrix-valued function $K(\xi)$, $\norm{K(\cdot)}_\infty \triangleq \sup_{\xi \in [0,1]}\norm{K(\xi)}$. $C^k([0,T];\altmathcal{Z})$ ($k \in \{1, 2, \dots, \infty\}$) denotes the space of $k$-times continuously differentiable functions on $[0,T]$ with values in $\altmathcal{Z}$ (for $T = \infty$, the interval $[0,T]$ is identified with $\mathbb{R}_{\geq 0}$). Given two Hilbert spaces $\altmathcal{V}$ and $\altmathcal{W}$, $\mathscr{L}(\altmathcal{V};\altmathcal{W})$ denotes the spaces of linear operators from $\altmathcal{V}$ to $\altmathcal{W}$ (abbreviated $\mathscr{L}(\altmathcal{V})$ if $\altmathcal{V} = \altmathcal{W}$). 


\section{Model description and preliminaries}\label{sect:Problem}
This section is dedicated to introducing the considered family of semilinear single-track models, along with the main assumptions formulated about their dynamics. In particular, the governing equations of the model are presented in Section~\ref{sect:modelSed}, whereas mild structural assumptions are postulated in Section~\ref{sect:ass}, where some preliminary results are also collected.

\subsection{Model description}\label{sect:modelSed}
In the following, Section~\ref{sect:exampleEqs} reviews the governing equations of the semilinear single-track models to the extent that is necessary to understand the manuscript, whereas Section~\ref{sect:stateSpace} introduces a compact state-space representation more amenable to mathematical analysis. 

\subsubsection{Lateral vehicle dynamics with distributed tire friction dynamics}\label{sect:exampleEqs} 
As illustrated schematically in Figure~\ref{figureForcePostdoc}, this paper examines semilinear single-track models that govern the lateral dynamics of a road vehicle traveling at a constant cruising speed. The formulations presented here consider tire models with both a rigid and flexible carcass and are adapted from \textcite{SemilinearV}, but with some state variables restated in a nondimensional form that is more amenable to singular perturbation analysis. 
\begin{figure}
\centering
\includegraphics[width=0.8\linewidth]{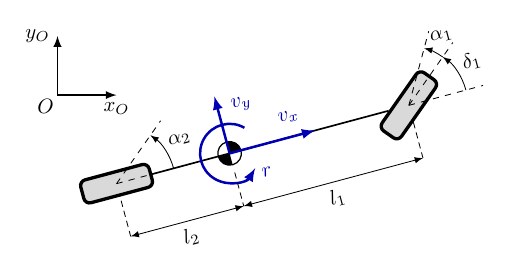} 
\caption{Single-track vehicle model.}
\label{figureForcePostdoc}
\end{figure}
In particular, for sufficiently small steering inputs, the linear ODE describing the rigid vehicle dynamics may be deduced to be \parencite{Guiggiani}
\begin{subequations}\label{eq:rigid}
\begin{align}
\dot{\beta}(t) & = -\dfrac{1}{mv_x}\bigl(F_{y1}(t) + F_{y2}(t)\bigr) -r(t), \\
\dot{r}(t) & = -\dfrac{1}{I_z}\bigl( l_1F_{y1}(t)-l_2F_{y2}(t)\bigr), && t\in (0,T),
\end{align}
\end{subequations}
where the lumped states $\mathbb{R} \ni \beta(t) \triangleq v_y(t)/v_x$ and $r(t) \in \mathbb{R}$ are the vehicle sideslip angle and yaw rate, being $v_y(t) \in \mathbb{R}$ and $v_x\in \mathbb{R}_{>0}$ its lateral and constant longitudinal speeds\footnote{When studying single-track vehicle models, the assumption of a constant longitudinal speed is standard \parencite{Guiggiani}, and closely related to neglecting the longitudinal slips.}, $m\in \mathbb{R}_{>0}$ and $I_z\in \mathbb{R}_{>0}$ denote respectively the vehicle mass and moment of inertia of the center of gravity around the vertical axis, and $l_1$, $l_2 \in \mathbb{R}_{>0}$ are the front and rear axle lengths. In turn, adopting a distributed friction model \parencite{DistrLuGre,FrBD}, the tire forces $F_{y1}(t)$, $F_{y2}(t) \in \mathbb{R}$ may be calculated as
\begin{align}\label{eq:Fi}
\begin{split}
F_{yi}(t) & = F_{zi}\int_0^1 \bar{p}_i(\xi)\biggl[\bar{\sigma}_{0,i}z_i(\xi,t)+\bar{\sigma}_{1,i}\dod{z_i(\xi,t)}{t}\biggr] \dif \xi \\
& \quad + 2F_{zi}\bar{\sigma}_{2,i}\alpha_i\bigl(\beta(t),r(t), \delta_i(t)\bigr), \quad i \in \{1,2\},
\end{split} 
\end{align}
where the distributed state $z_i(\xi,t) \in \mathbb{R}$, $i = \{1,2\}$, represents the nondimensional deflection of a bristle schematizing a rubber element within the contact patch, $\alpha_i(\beta(t),r(t), \delta_i(t)) \in \mathbb{R}$ the corresponding apparent slip angle, $F_{zi} \in \mathbb{R}_{>0}$ denotes the vertical force acting on the tire, $\bar{p}_i \in C^1([0,1];\mathbb{R}_{\geq 0})$ is the nondimensional vertical pressure distribution, and
\begin{subequations}
\begin{align}
\bar{\sigma}_{0,i} & \triangleq \sigma_{0,i}L_i, \\
\bar{\sigma}_{1,i} & \triangleq \sigma_{1,i}L_i, \\
 \bar{\sigma}_{2,i} & \triangleq \sigma_{2,i}v_x, \quad i \in \{1,2\},
\end{align}
\end{subequations}
being $\sigma_{0,i} \in \mathbb{R}_{>0}$ the normalized micro-stiffness coefficient, $\sigma_{1,i} \in \mathbb{R}_{\geq 0}$ the normalized micro-damping coefficient, and $\sigma_{2,i}\in \mathbb{R}_{\geq 0}$ the viscous coefficient \parencite{DistrLuGre}. In turn, for a tire with a rigid carcass, the bristle dynamics obeys the following semilinear PDE \parencite{SemilinearV}:
\begin{subequations}\label{Eq:PDEz}
\begin{align}\label{Eq:PDEzDKKF}
\begin{split}
& \dfrac{\bar{L}}{v_x}\dpd{z_i(\xi,t)}{t} + \dfrac{1}{\bar{L}_i}\dpd{z_i(\xi,t)}{\xi} =-\dfrac{1}{\bar{L}_i}\dfrac{\bar{\sigma}_{0,i}\abs{\alpha_i\bigl(\beta(t),r(t), \delta_i(t)\bigr)}_\varepsilon}{\bar{g}_i\bigl(\alpha_i\bigl(\beta(t),r(t), \delta_i(t)\bigr)\bigr)}\\
& \quad \times z_i(\xi,t)  +\dfrac{2}{\bar{L}_i}\dfrac{\bar{\mu}_i\bigl(\alpha_i\bigl(\beta(t),r(t), \delta_i(t)\bigr)\bigr)}{\bar{g}_i\bigl(\alpha_i\bigl(\beta(t),r(t), \delta_i(t)\bigr)\bigr)} \alpha_i\bigl(\beta(t),r(t), \delta_i(t)\bigr), \\
& \qquad \qquad \qquad \qquad \qquad \qquad (\xi,t) \in (0,1)\times (0,T),
\end{split} \\
& z_i(0,t) = 0, \label{eq:BCz}
\end{align}
\end{subequations}
with
\begin{align}
\bar{L}_i \triangleq \dfrac{L_i}{\bar{L}}, \quad i \in \{1,2\},
\end{align}
where $L_i \in \mathbb{R}_{>0}$ indicates the contact patch length, $\bar{\mu}_i \in C^0(\mathbb{R};[\mu\ped{min},\infty))$, with $\mu\ped{min}\in \mathbb{R}_{>0}$, the friction coefficient as a function of the slip angle, and the function $\abs{\cdot}_\varepsilon \in C^0(\mathbb{R};\mathbb{R}_{\geq 0})$ denotes the (possibly regularized\footnote{It is common in engineering practice to replace the absolute value with differentiable functions $\abs{\cdot}_\varepsilon \in C^1(\mathbb{R};\mathbb{R}_{\geq 0})$ \parencite{Rill,Rill0}, such as $\abs{v}_\varepsilon \triangleq \sqrt{v^2 + \varepsilon}$, for some $\varepsilon \in \mathbb{R}_{>0}$. This paper considers $\varepsilon \in \mathbb{R}_{\geq 0}$.}) absolute value, converging uniformly to $\abs{\cdot}$ in $C^0(\mathbb{R};\mathbb{R}_{\geq 0})$ for $\varepsilon \to 0$. Distributed friction models accommodated by~\eqref{eq:Fi} and~\eqref{Eq:PDEz} include the Dahl model, as well as the LuGre and FrBD formulations with the damping term modeled as a linear function of the total time derivative of the bristle deformation\footnote{Model variants replacing the total time derivative with the partial one in the definition of the damping term in~\eqref{eq:Fi} do not fit the proposed singular perturbation framework. In this context, it is, however, worth mentioning that the partial time derivative does not represent a real deformation velocity, as opposed to the total one. Besides, there are also several mathematical arguments in favor of the adoption of the total time derivative, as extensively discussed in \textcite{DistrLuGre,FrBD}.}.

Alternatively, for a tire with a flexible tire carcass, and $\bar{\sigma}_{1,i} = \bar{\sigma}_{2,i} = 0$, $i \in \{1,2\}$, in~\eqref{eq:Fi} (which also implies $\bar{g}_i(\cdot) = \bar{\mu}_i(\cdot)$, see \textcite{SemilinearV}), the PDE governing the bristle dynamics may be deduced as follows: 
\begin{subequations}\label{Eq:PDEzFlex}
\begin{align}\label{Eq:PDEzDKKFFlex}
\begin{split}
& \dfrac{\bar{L}}{v_x}\dpd{z_i(\xi,t)}{t} + \dfrac{1}{\bar{L}_i}\dpd{z_i(\xi,t)}{\xi} =-\dfrac{1}{\bar{L}_i}\dfrac{\bar{\sigma}_{0,i}\abs{\alpha_i\bigl(\beta(t),r(t), \delta_i(t)\bigr)}_\varepsilon}{\bar{\mu}_i\bigl(\alpha_i\bigl(\beta(t),r(t), \delta_i(t)\bigr)\bigr)}\\
& \quad \times \Biggl(z_i(\xi,t)-\psi_i\int_0^1\bar{p}_i(\xi)z_i(\xi) \dif \xi\Biggr)  \\
& \quad +\dfrac{\psi_i}{\bar{L}_i}\Biggl( \bar{p}_i(1)z_i(1) - \int_0^1 \dod{\bar{p}_i(\xi)}{\xi}z_i(\xi)\dif \xi\Biggr) \\
& \quad +2 \dfrac{\phi_i}{\bar{L}_i} \alpha_i\bigl(\beta(t),r(t), \delta_i(t)\bigr), \quad (\xi,t) \in (0,1)\times (0,T),
\end{split} \\
& z_i(0,t) = 0, \label{eq:BCzFlex}
\end{align}
\end{subequations}
where the constants $\phi_i \in (0,1]$ and $\psi_i \in [0,1)$ are structural parameters connected with the flexibility of the tire carcass, identically satisfying $\phi_i + \psi_i = 1$, $i \in \{1,2\}$.

Finally, the parameter $\bar{L}\in \mathbb{R}_{>0}$ appearing in both~\eqref{Eq:PDEzDKKF} and~\eqref{Eq:PDEzDKKFFlex} represents a characteristic length of the problem~\eqref{Eq:PDEz}, which may be defined as $\bar{L} \triangleq (L_1+L_2)/2$ or $\bar{L} \triangleq \max\{L_1,L_2\}$.

The apparent slip angles in~\eqref{eq:Fi},~\eqref{Eq:PDEzDKKF}, and~\eqref{Eq:PDEzDKKFFlex} are given by
\begin{subequations}\label{eq:slipAngles}
\begin{align}
\begin{split}
\alpha_1\bigl(\beta(t),r(t), \delta_1(t)\bigr) & = \beta(t) + \dfrac{l_1}{v_x}r(t)-\delta_1(t), 
\end{split} \\
\begin{split}
\alpha_2\bigl(\beta(t),r(t), \delta_2(t)\bigr)& =  \beta(t)-\dfrac{l_2}{v_x}r(t)-\chi \delta_2(t),
\end{split}
\end{align}
\end{subequations}
where the steering angles at the front and rear axles, $\delta_1(t), \delta_2(t) \in \mathbb{R}$ respectively, represent the control inputs, and $\chi\in \{0,1\}$ a parameter describing the actuation at the rear wheels.

Equations~\eqref{eq:rigid}-\eqref{eq:slipAngles} completely determine the lateral motion of the single-track model. As explained next in Section~\ref{sect:stateSpace}, they may be restated in a compact form which is more suited to mathematical analysis.

\subsubsection{State-space representation}\label{sect:stateSpace} 
The objective of this paper consists of studying the (local) stability and stabilization of the semilinear single-track models introduced in Section~\ref{sect:exampleEqs}, by exploiting the time-scale separation between the slow ODE and fast PDE subsystems, respectively. Indeed, at sufficiently high longitudinal speeds $v_x$, the tire dynamics evolve much faster than the rigid body ones. This informal argument has traditionally motivated studying the stability of road vehicles using reduced-order descriptions that approximate the transient tire forces with their steady-state solutions. In a similar context, controllers and state observers are typically designed by neglecting tire dynamics. 

In this paper, the time-scale parameter $\tau_z \in \mathbb{R}_{>0}$ is introduced as
\begin{align}\label{eq:epsilonDef}
\tau_z \triangleq \dfrac{\bar{L}}{v_x},
\end{align}
and represents the average/maximum time required for tire particles to traverse the contact patches of the front and rear axles. In this context, it is worth emphasizing that the precise definition of $\bar{L}$ in~\eqref{eq:epsilonDef} is of limited significance for the scope of this paper, since the tire contact patch lengths are comparable under normal operating conditions. Further considerations on this aspect are provided below.
\begin{remark}
Typically, $\bar{L} <1$ may be treated as a fixed parameter under certain operating conditions (although it may depend on inflation pressure, vehicle weight, etc.), whereas $v_x$ may vary over a broad range. Consequently, the time constant $\tau_z$ mostly depends on the vehicle's cruising speed. In this context, for low values of $v_x$, $\tau_z$ may become relatively large (see also Remarks~\ref{remark:tveh} and~\ref{remark:tveh2}), causing tire transients to dominate over steady-state dynamics and thus compromising the validity of the perturbation approach, as discussed more extensively in Sections~\ref{sect:stabl} and~\ref{sect:Controllllll}.
\end{remark}

Using~\eqref{eq:epsilonDef}, and defining $\mathbb{R}^2 \ni X(t) \triangleq [\beta(t)\; r(t)]^{\mathrm{T}}$, $\mathbb{R}^2\ni z(\xi,t) \triangleq [z_1(\xi,t) \; z_2(\xi,t)]^{\mathrm{T}}$, $\mathbb{R}^2 \ni U(t) \triangleq [\delta_1(t)\; \delta_2(t)]^{\mathrm{T}}$, and $\mathbb{R}^2 \ni \alpha(X(t),U(t)) = [\alpha_1(X(t),U(t))\; \alpha_2(X(t),U(t))]^{\mathrm{T}} \triangleq [\alpha_1(\beta(t),r(t), \delta_1(t)) \; \alpha_2(\beta(t),r(t), \delta_2(t))]^{\mathrm{T}}$, both the model variants with rigid and flexible tire carcass may be recast in the form 
\begin{subequations}\label{eq:originalSystems}
\begin{align}
\begin{split}
& \dot{X}(t) =A_1X(t)+ G_1\Bigl[(\mathscr{K}_1z)(t)+\Sigma\bigl(\alpha\bigl(X(t),U(t)\bigr)\bigr)(\mathscr{K}_2z)(t)\Bigr] \\
& \qquad \qquad+G_1h_1\bigl(\alpha(X(t),U(t))\bigr),\quad t \in (0,T),
\end{split} \label{eq:originalSystemsODE}\\
\begin{split}
& \tau_z\dpd{z(\xi,t)}{t} + \Lambda \dpd{z(\xi,t)}{\xi} = \Sigma\Bigl(\alpha\bigl(X(t),U(t)\bigr)\Bigr) \\
& \qquad \qquad \times  \bigl[z(\xi,t) + (\mathscr{K}_3z)(t)\bigr]  +(\mathscr{K}_4z)(t) \\
& \qquad \qquad  +h_2\Bigl(\alpha\bigl(X(t),U(t)\bigr)\Bigr),\quad (\xi,t) \in (0,1) \times (0,T),
\end{split} \label{eq:originalSystemsPDE} \\
&z(0,t) = 0, \quad t \in (0,T),\label{eq:originalSystemsBC}
\end{align}
\end{subequations}
where the apparent slip angle vector $\alpha \in C^1(\mathbb{R}^{4};\mathbb{R}^{2})$ may be expressed as
\begin{align}\label{eq:relVel}
\alpha(X,U) \triangleq A_2X + G_2U.
\end{align}
In~\eqref{eq:originalSystems} and~\eqref{eq:relVel}, the matrix $\mathbf{GL}_{2}(\mathbb{R})\cap \mathbf{Sym}_{2}(\mathbb{R})\ni \Lambda \succ 0$ collects the transport velocities, $\Sigma \in C^0(\mathbb{R}^{2};\mathbf{M}_{2}(\mathbb{R}))$ represents the nonlinear source matrix, $A_1, A_2, G_1, G_2 \in  \mathbf{M}_{2}(\mathbb{R})$ are matrices with constant coefficients, and $h_1, h_2 \in C^0(\mathbb{R}^{2 };\mathbb{R}^{2 })$ are vector-valued functions. Finally, the operators $(\mathscr{K}_1\zeta)$, $(\mathscr{K}_2\zeta)$, $(\mathscr{K}_3\zeta)$, and $(\mathscr{K}_4\zeta)$ satisfy $\mathscr{K}_1, \mathscr{K}_2, \mathscr{K}_3\in \mathscr{L}(L^2((0,1);\mathbb{R}^{2});\mathbb{R}^{2})$, and $\mathscr{K}_4\in \mathscr{L}(H^1((0,1);\mathbb{R}^{2});\mathbb{R}^{2})$, with
\begin{subequations}\label{eq:operatorK}
\begin{align}
(\mathscr{K}_i\zeta) &\triangleq \int_0^1 K_i(\xi)  \zeta(\xi) \dif \xi, \quad i \in \{1,2,3\},\\
(\mathscr{K}_4\zeta) &\triangleq \int_0^1 K_4(\xi)  \zeta(\xi) \dif \xi+K_5\zeta(1),
\end{align}
\end{subequations}
where $K_1,K_2, K_3,K_4\in C^0([0,1];\mathbf{M}_{2}(\mathbb{R}))$, and $K_5 \in \mathbf{M}_{2}(\mathbb{R})$.
Equations~\eqref{eq:originalSystems}-\eqref{eq:operatorK} accommodate semilinear single-track vehicle models with both a rigid and flexible tire carcass, as formalized below in Parametrizations~\ref{param:1} and~\ref{param:2}.

\begin{parametrization}[Semilinear single-track models with a rigid tire carcass]\label{param:1}
Semilinear single-track models with a rigid carcass admit a state-space representation in the form described by~\eqref{eq:originalSystems}-\eqref{eq:operatorK}, with \begin{small}
\begin{align}\label{eq:MatricesOOO}
A_1 & \triangleq \begin{bmatrix} 0 & -1 \\ 0 &  0 \end{bmatrix}, \quad G_1 \triangleq -\begin{bmatrix}\dfrac{1}{mv_x} & \dfrac{1}{mv_x} \\ \dfrac{l_1}{I_z} & -\dfrac{l_2}{I_z} \end{bmatrix}, \nonumber\\
K_1(\xi) & \triangleq \begin{bmatrix} F_{z1}\bar{\sigma}_{0,1}\bar{p}_1(\xi) & 0 \\ 0 & F_{z2}\bar{\sigma}_{0,2}\bar{p}_2(\xi)\end{bmatrix}, \nonumber\\
K_2(\xi) & \triangleq \begin{bmatrix} F_{z1}\bar{\sigma}_{1,1}\bar{p}_1(\xi) & 0 \\ 0 & F_{z2}\bar{\sigma}_{1,2}\bar{p}_2(\xi)\end{bmatrix}, \nonumber\\
\Lambda & \triangleq \begin{bmatrix} \dfrac{1}{\bar{L}_1} & 0 \\ 0 & \dfrac{1}{\bar{L}_2}\end{bmatrix}, \quad \Sigma(\alpha)  \triangleq \begin{bmatrix} -\dfrac{\bar{\sigma}_{0,1}\abs{\alpha_1}_\varepsilon}{\bar{L}_1\bar{g}_1(\alpha_1)} & 0 \\0 & -\dfrac{\bar{\sigma}_{0,2}\abs{\alpha_2}_\varepsilon}{\bar{L}_2\bar{g}_2(\alpha_2)} \end{bmatrix}, \nonumber\\
h_1(\alpha) & \triangleq 2\begin{bmatrix} F_{z1}\biggl(\bar{\sigma}_{1,1}\dfrac{\bar{\mu}_1(\alpha_1)}{\bar{L}_1\bar{g}_1(\alpha_1)} + \bar{\sigma}_{2,1}\biggr)\alpha_1  \\ F_{z2}\biggl(\bar{\sigma}_{1,2}\dfrac{\bar{\mu}_2(\alpha_2)}{\bar{L}_2\bar{g}_2(\alpha_2)} + \bar{\sigma}_{2,2}\biggr)\alpha_2 \end{bmatrix}, \nonumber \\
h_2(\alpha) & \triangleq 2\begin{bmatrix} \dfrac{\bar{\mu}_1(\alpha_1)}{\bar{L}_1\bar{g}_1(\alpha_1)}\alpha_1 \\ \dfrac{\bar{\mu}_2(\alpha_2)}{\bar{L}_2\bar{g}_2(\alpha_2)}\alpha_2 \end{bmatrix}, 
\end{align}\end{small}
$K_3 = K_4 = K_5=0$, and the matrices in~\eqref{eq:relVel} reading
\begin{align}\label{eq:A2G2}
A_2 & \triangleq \begin{bmatrix}1 & \dfrac{l_1}{v_x} \\ 1 & -\dfrac{l_2}{v_x} \end{bmatrix}, \quad G_2 \triangleq -\begin{bmatrix} 1 & 0 \\ 0 & \chi\end{bmatrix}.
\end{align} 
\end{parametrization}

\begin{parametrization}[Semilinear single-track models with a flexible tire carcass]\label{param:2}
Semilinear single-track models with a flexible tire carcass may be put compactly in the form described by~\eqref{eq:originalSystems}-\eqref{eq:operatorK}, with $A_1$, $G_1$, $K_1$, and $\Lambda$ as in~\eqref{eq:MatricesOOO}, $K_2 =0$, $h_1(\alpha) = 0$, $A_2$ and $G_2$ according to~\eqref{eq:A2G2}, and \begin{small}
\begin{align}\label{eq:MatricesOOO2}
K_3(\xi) & \triangleq -\begin{bmatrix} \psi_{1}\bar{p}_1(\xi) & 0 \\ 0 & \psi_{2}\bar{p}_2(\xi)\end{bmatrix}, \nonumber\\
 K_4(\xi) & \triangleq -\begin{bmatrix} \dfrac{\psi_1}{\bar{L}_1}\dod{\bar{p}_1(\xi)}{\xi} & 0 \\ 0 & \dfrac{\psi_2}{\bar{L}_2}\dod{\bar{p}_2(\xi)}{\xi}\end{bmatrix},  \nonumber\\
 K_5 & \triangleq \begin{bmatrix} \dfrac{\psi_1}{\bar{L}_1}\bar{p}_1(1) & 0 \\ 0 & \dfrac{\psi_2}{\bar{L}_2}\bar{p}_2(1)\end{bmatrix},  \nonumber\\
\Sigma(\alpha) & \triangleq \begin{bmatrix} -\dfrac{\bar{\sigma}_{0,1}\abs{\alpha_1}_\varepsilon}{\bar{L}_1\bar{\mu}_1(\alpha_1)} & 0 \\0 & -\dfrac{\bar{\sigma}_{0,2}\abs{\alpha_2}_\varepsilon}{\bar{L}_2\bar{\mu}_2(\alpha_2)} \end{bmatrix}, \nonumber \\
h_2(\alpha) &  \triangleq 2\begin{bmatrix}\dfrac{\phi_1}{\bar{L}_1} & 0 \\ 0 & \dfrac{\phi_2}{\bar{L}_2}\end{bmatrix}\alpha.
\end{align}\end{small}
\end{parametrization}
From a mathematical perspective, the semilinear hyperbolic ODE-PDE interconnection~\eqref{eq:originalSystems} is (locally) well-posed. In particular, this paper considers the Hilbert space $\altmathcal{X}\triangleq \mathbb{R}^{2}\times L^2((0,1);\mathbb{R}^{2})$, equipped with norm $\norm{(Z, \zeta(\cdot))}_{\altmathcal{X}}^2 \triangleq \norm{Z}_{2}^2 + \norm{\zeta(\cdot)}_{L^2((0,1);\mathbb{R}^{2})}^2$. Accordingly, if $\Sigma \in C^0(\mathbb{R}^{2};\mathbf{M}_{2}(\mathbb{R}))$ and $h_1,h_2 \in C^0(\mathbb{R}^2;\mathbb{R}^2)$ are locally Lipschitz continuous, the existence and uniqueness of \emph{mild solution} $(X,z) \in C^0([0,t\ped{max});\altmathcal{X})$ follows from~\textcite[Theorem~3.3]{SemilinearV} for all inputs $U \in C^0([0,T];\mathbb{R}^{2})$ and initial conditions (ICs) $(X_0,z_0) \triangleq (X(0),z(\cdot,0)) \in \altmathcal{X}$.

A schematic of the ODE-PDE system~\eqref{eq:originalSystems} is illustrated in Figure~\ref{figure:system0}, where, for convenience of notation, $\mathbb{R}^{2} \ni W_1(t) \triangleq (\mathscr{F}_1(X,U,z))(t)$ and $\mathbb{R}^{2} \ni W_2(\xi,t) \triangleq (\mathscr{F}_2(X,U,z))(\xi,t)$, with $ (\mathscr{F}_1(X,U,z))(t) \triangleq G_1(\mathscr{K}_1z)(t) + G_1\Sigma(\alpha(X(t),U(t)))(\mathscr{K}_2z)(t)  + G_1h_1(\alpha(X(t),U(t)))$ and $(\mathscr{F}_2(X,U,z))(\xi,t)\triangleq \Sigma(\alpha(X(t),U(t)))[z(\xi,t)+(\mathscr{K}_3z)(t)] + (\mathscr{K}_4z)(t)+h_2(\alpha(X(t),U(t)))$.

\begin{figure}
\centering
\includegraphics[width=0.9\linewidth]{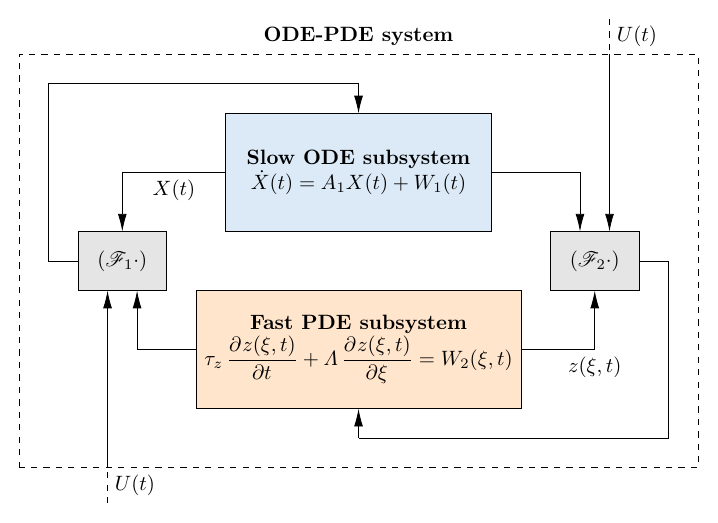} 
\caption{Schematic representation of the ODE-PDE interconnection~\eqref{eq:originalSystems}.}
\label{figure:system0}
\end{figure}

\subsection{Assumptions and preliminaries}\label{sect:ass}
This Section introduces the main assumptions required to synthesize the proposed controller.

\subsubsection{Structural assumptions}
The structural assumptions formulated in this paper concern the fast PDE subsystem~\eqref{eq:originalSystemsPDE}, and reflect the fact that rolling friction is a dissipative phenomenon.
\begin{assumption}[Strict dissipativity]\label{ass:ODEf} There exists $C^0([0,1]; \mathbf{Sym}_{2}(\mathbb{R})) \ni \mathscr{Q} \triangleq \diag\{\mathscr{Q}_1,\mathscr{Q}_2\}$, with $\mathscr{Q}(\xi) \succ 0$, such that the unbounded operator $(\mathscr{A},\mathscr{D}(\mathscr{A}))$, defined by
\begin{subequations}\label{eq:operatorA}
\begin{align}
(\mathscr{A}\zeta)(\xi) & \triangleq -\Lambda\dpd{\zeta(\xi)}{\xi} + (\mathscr{K}_4\zeta)(\xi), \\
\mathscr{D}(\mathscr{A}) & \triangleq \Bigl\{ \zeta \in H^1((0,1);\mathbb{R}^{2}) \mathrel{\Big|} \zeta(0) = 0\Bigr\}, 
\end{align}
\end{subequations}
satisfies
\begin{align}\label{eq:dissA000}
\begin{split}
& \RE\langle \mathscr{A}\zeta,\mathscr{Q}\zeta\rangle_{L^2((0,1);\mathbb{R}^{2})} \leq -\omega \norm{\zeta(\cdot)}_{L^2((0,1);\mathbb{R}^{2})}^2,
\end{split}
\end{align}
for all $\zeta \in \mathscr{D}(\mathscr{A})$ and some $\omega \in \mathbb{R}_{>0}$.
\end{assumption}
Assumption~\ref{ass:ODEf} implies that $(\mathscr{A},\mathscr{D}(\mathscr{A}))$ generates an exponentially stable $C_0$-semigroup on $L^2((0,1);\mathbb{R}^{2})$. Its physical interpretation is that rolling contact processes, when considered in isolation, are stable (in fact, they are even dissipative).

\begin{assumption}[Dissipativity and Lipschitz continuity]\label{ass:ODEf2}
For every $(y,\zeta) \in \mathbb{R}^{2 }\times L^2((0,1);\mathbb{R}^{2 })$, the matrix $\Sigma(y) \in \mathbf{M}_{2}(\mathbb{R})$ satisfies
\begin{align}\label{eq:ineqP}
& \int_0^1 \zeta^{\mathrm{T}}(\xi) \mathscr{Q}(\xi)\Sigma(y)\bigl[\zeta(\xi) + (\mathscr{K}_3\zeta)\bigr] \dif \xi \leq 0.
\end{align}
Additionally, it is globally Lipschitz continuous, that is, there exists $L_{\Sigma}\in \mathbb{R}_{\geq 0}$ such that
\begin{align}\label{eq:SigmaLSigma}
\norm{\Sigma(y_1)-\Sigma(y_2)} & \leq L_{\Sigma}\norm{y_1-y_2}_2, 
\end{align}
for all $y_1,y_2 \in \mathbb{R}^{2}$. 
\end{assumption}

Some comments are in order. First, the dissipativity inequalities~\eqref{eq:dissA000} and~\eqref{eq:ineqP} in Assumptions~\ref{ass:ODEf} and~\ref{ass:ODEf2} are trivial to prove for $\psi_i = 0$, $i \in \{1,2\}$. More generally, the bounds~\eqref{eq:dissA000} and~\eqref{eq:ineqP} may be verified for several combinations of model parameters, for instance by resorting to passivity arguments, and choosing the matrix $\mathscr{Q}(\xi)$ as a multiple of $K_1(\xi)$ (the reader may consult \textcite{PassExp} for further details on this matter). Finally, the Lipschitz condition~\eqref{eq:SigmaLSigma} holds for the Dahl, LuGre, and FrBD models, as well as all the formulations considered in \textcite{SemilinearV}. In fact, in some cases, the matrix-valued function $\Sigma \in C^1(\mathbb{R}^2;\mathbf{M}_2(\mathbb{R}))$ is uniformly bounded, which may ensure global well-posedness for the mild solutions of the open-loop ODE-PDE system~\eqref{eq:originalSystems} (see \textcite{SemilinearV}). 
Together, Assumptions~\ref{ass:ODEf} and~\ref{ass:ODEf2} permit recovering some preliminary results, as formalized next in Proposition~\ref{prop2:strictDDD} and Lemma~\ref{prop:pass}.

\begin{proposition}[Strict dissipativity]\label{prop2:strictDDD}
Suppose that Assumptions~\ref{ass:ODEf} and~\ref{ass:ODEf2} hold. Then, for all $y \in \mathbb{R}^{2}$, the unbounded operator $(\mathscr{A}_{\Sigma}(y),\mathscr{D}(\mathscr{A}_{\Sigma}(y)))$, defined by
\begin{subequations}\label{eq:Ay}
\begin{align}
(\mathscr{A}_{\Sigma}(y)\zeta)(\xi) & \triangleq  \Sigma(y)\bigl[\zeta(\xi) + (\mathscr{K}_3\zeta)\bigr] +(\mathscr{A}\zeta)(\xi), \\
\mathscr{D}(\mathscr{A}_{\Sigma}(y)) & = \mathscr{D}(\mathscr{A}) \triangleq \Bigl\{ \zeta \in H^1((0,1);\mathbb{R}^{2}) \mathrel{\Big|} \zeta(0) = 0\Bigr\}, 
\end{align}
\end{subequations}
also generates an exponentially stable $C_0$-semigroup on $L^2((0,1);\mathbb{R}^{2})$. In particular,
\begin{align}\label{eq:dissA0001}
\begin{split}
& \RE\langle \mathscr{A}_{\Sigma}(y)\zeta,\mathscr{Q}\zeta\rangle_{L^2((0,1);\mathbb{R}^{2})}  \leq -\omega \norm{\zeta(\cdot)}_{L^2((0,1);\mathbb{R}^{2})}^2,
\end{split}
\end{align}
for all $(y,\zeta) \in \mathbb{R}^{2} \times \mathscr{D}(\mathscr{A}_{\Sigma}(y))$, with the same $\omega \in \mathbb{R}_{>0}$ as in~\eqref{eq:dissA000}. 
\begin{proof}
The assertion is an immediate consequence of the inequalities~\eqref{eq:dissA000} and~\eqref{eq:ineqP}.
\end{proof}
\end{proposition}

\begin{lemma}\label{prop:pass}
Suppose that Assumptions~\ref{ass:ODEf} and~\ref{ass:ODEf2} hold. Then, for every $\mathbb{R}^{2} \ni y = [y_1\; y_2]^{\mathrm{T}}$, the nonlocal ODE
\begin{subequations}\label{eq:PDEphi}
\begin{align}
\begin{split}
\Lambda \dpd{\hat{z}(\xi,y)}{\xi} & = \Sigma(y)\bigl[\hat{z}(\xi,y) + (\mathscr{K}_3\hat{z})(y)\bigr]  +(\mathscr{K}_4\hat{z})(y)\\
& \quad +h_2(y),\quad \xi \in (0,1) \label{eq:originalSystemsPDEODE} 
\end{split} \\
\hat{z}(0,y) & = 0, \label{eq:originalSystemsBCODE}
\end{align}
\end{subequations}
admits a unique solution $\hat{z}(\cdot,y) \in C^1([0,1];\mathbb{R}^2)$, with $\mathbb{R}^2 \ni \hat{z}(\xi,y) = [\hat{z}_1(\xi,y_1)\; \hat{z}_2(\xi,y_2)]^{\mathrm{T}}$.
\begin{proof}
The result is a consequence of inequality~\eqref{eq:dissA0001}.
\end{proof}
\end{lemma}
For what follows, it is also beneficial to define the function $\mathbb{R}^2 \ni \hat{F}_y(\cdot) \triangleq [\hat{F}_{y1}(\cdot) \; \hat{F}_{y2}(\cdot)]^{\mathrm{T}}$ as
\begin{align}\label{eq:functionPhi}
\begin{split}
 \hat{F}_y(y) & \triangleq (\mathscr{K}_1\hat{z})(y)+\Sigma(y)(\mathscr{K}_2\hat{z})(y)+h_1(y).
\end{split}
\end{align}
Utilizing~\eqref{eq:PDEphi} and~\eqref{eq:functionPhi}, the equilibria $(X^\star,z^\star)\in \mathbb{R}^2\times\mathscr{D}(\mathscr{A})$ of~\eqref{eq:originalSystems} associated with a constant input $U^\star \in \mathbb{R}^2$ are given by
\begin{subequations}\label{eq:equilibrium}
\begin{align}
A_1X^\star + G_1F_y^\star & = 0, \\
F_y^\star - \hat{F}_y(\alpha^\star) & = 0, \\
z^\star(\xi) - \hat{z}(\xi,\alpha^\star) & = 0, \\
\alpha^\star - A_2X^\star - G_2U^\star & = 0,
\end{align}
\end{subequations}
where $\mathbb{R}^2 \ni F_y^\star = [F_{y1}^\star\; F_{y2}^\star]^{\mathrm{T}}$ and $\mathbb{R}^2 \ni \alpha^\star = [\alpha_1^\star\; \alpha_2^\star]^{\mathrm{T}}$ denote respectively the steady-state tire forces and apparent slip angles.

In particular, the following Assumption~\ref{ass:differ} is postulated concerning the functions $\hat{z}(\xi,\cdot)$ and $\hat{F}_y(\cdot)$.
\begin{assumption}\label{ass:differ}
The functions $\hat{z}(\xi,\cdot) \in C^1([0,1]\times\mathbb{R}^2;\mathbb{R}^2)$, $\hat{F}_y\in C^1(\mathbb{R}^2; \mathbb{R}^2)$. Moreover, there exists $M_{\hat{z}}\in \mathbb{R}_{\geq 0}$ such that
\begin{align}
\max_{y \in \mathbb{R}^2}\norm{\dpd{\hat{z}(\cdot,y)}{y}}_\infty \leq M_{\hat{z}}.
\end{align}
\end{assumption}
Assumption~\ref{ass:differ} is always verified in practice. For instance, Figure~\ref{figure:typicalForce} illustrates some typical trends for the function $\hat{F}_y(\cdot)$, obtained for constant and exponentially decreasing pressure distributions $\bar{p}_i(\xi)$, $i \in \{1,2\}$.

Before moving to the next Section~\ref{sect:assOBS}, it is useful to introduce the matrix $\tilde{C}(\alpha^\star) \in \mathbf{M}_2(\mathbb{R})$ of generalized cornering stiffnesses as
\begin{align}\label{eq:matrixCorn}
\begin{split}
\tilde{C}(\alpha^\star) & = \begin{bmatrix} \tilde{C}_1(\alpha_1^\star) & 0 \\ 0 & \tilde{C}_2(\alpha_2^\star) \end{bmatrix} \\
& \triangleq \begin{bmatrix} \eval{\dod{\hat{F}_{y1}(y_1)}{y_1}}_{y_1=\alpha_1^\star} & 0 \\ 0 &  \eval{\dod{\hat{F}_{y2}(y_2)}{y_2}}_{y_2 = \alpha_2^\star} \end{bmatrix} \\
& = \eval{\nabla_{y}\hat{F}_y(y)^{\mathrm{T}}}_{y = \alpha^\star}.
\end{split}
\end{align}
In the vehicle dynamics literature, the quantities $\tilde{C}_1(\alpha_1^\star), \tilde{C}_2(\alpha_2^\star) \in \mathbb{R}$ appearing in~\eqref{eq:matrixCorn} are typically referred to as \emph{generalized cornering stiffnesses} of the front and rear axle, respectively.

\begin{figure}
\centering
\includegraphics[width=0.9\linewidth]{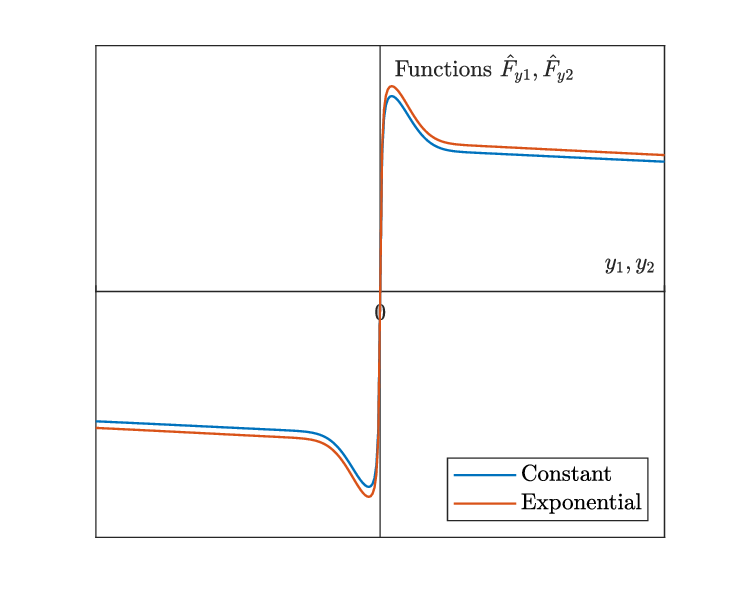} 
\caption{Typical trends of the function $\hat{F}_y(\cdot)$ obtained for constant and exponentially decreasing pressure distributions $\bar{p}_i(\xi)$, $i \in \{1,2\}$.}
\label{figure:typicalForce}
\end{figure}


\subsubsection{Assumptions for controller design}\label{sect:assOBS}
A last set of assumptions is introduced to facilitate the design of a controller that locally stabilizes the ODE-PDE system~\eqref{eq:originalSystems}. In this context, it is necessary to define the matrices $\tilde{A}_1(\alpha^\star),\tilde{G}_1(\alpha^\star) \in \mathbf{M}_2(\mathbb{R})$:
\begin{subequations}
\begin{align}
\begin{split}
& \tilde{A}_1(\alpha^\star) \triangleq A_1 + G_1\tilde{C}(\alpha^\star)A_2 \\
& = -\begin{bmatrix}\dfrac{\tilde{C}_1(\alpha_1^\star) + \tilde{C}_2(\alpha_2^\star)}{mv_x} & \dfrac{l_1\tilde{C}_1(\alpha_1^\star) - l_2\tilde{C}_2(\alpha_2^\star)}{mv_x^2}+1 \\
\dfrac{l_1\tilde{C}_1(\alpha_1^\star) - l_2\tilde{C}_2(\alpha_2^\star)}{I_z} & \dfrac{l_1^2\tilde{C}_1(\alpha_1^\star) + l_2^2\tilde{C}_2(\alpha_2^\star)}{I_zv_x}
\end{bmatrix},
\end{split} \label{eq:A1tilde} \\
\begin{split}
& \tilde{G}_1(\alpha^\star) \triangleq G_1\tilde{C}(\alpha^\star)G_2= \begin{bmatrix}\dfrac{\tilde{C}_1(\alpha_1^\star)}{mv_x} & \chi\dfrac{\tilde{C}_2(\alpha_2^\star)}{mv_x} \\
\dfrac{l_1\tilde{C}_1(\alpha_1^\star)}{I_z} & -\chi\dfrac{l_2\tilde{C}_2(\alpha_2^\star)}{I_z} 
\end{bmatrix}.
\end{split}
\end{align}
\end{subequations}
Accordingly, stabilizability assumptions are formulated below.
\begin{assumption}[Stabilizability]\label{ass:stabilizability}
The pair $(\tilde{A}_1(\alpha^\star), \tilde{G}_1(\alpha^\star))$ is stabilizable.
\end{assumption}

\begin{remark}\label{remark:assDet}
Typically, $\tilde{C}_1(\alpha_1^\star) > 0$ suffices to guarantee the \emph{controllability} of the pair $(\tilde{A}_1(\alpha^\star), \tilde{G}_1(\alpha^\star))$ in Assumption~\ref{ass:stabilizability}. Therefore, excluding maneuvers at the limit of handling where $\tilde{C}_1(\alpha_1^\star) = \tilde{C}_2(\alpha_2^\star)=0$, Assumption~\ref{ass:stabilizability} is always fulfilled in normal operating conditions of the vehicle.
\end{remark}

\section{Stability}\label{sect:stabl}
The present Section investigates the stability of the ODE-PDE system~\eqref{eq:originalSystems}. Specifically, Section~\ref{sect:StateFedd0} adopts a singular perturbation approach, considering the reduced and boundary layer subsystems in isolation, whereas Section~\ref{sect:StateFedde} recovers a Tikhonov-like stability result for sufficiently small values of the parameter $\tau_z$.

\subsection{Analysis via singular\\ perturbation theory}\label{sect:StateFedd0}
By exploiting the timescale separation between the ODE and PDE equations, the stability of the interconnection~\eqref{eq:originalSystems} may be studied by considering the corresponding reduced ODE and boundary layer PDE subsystems, as separately done in Sections~\ref{eq:reducedODEe} and~\ref{eq:BoundaryLayere}, respectively.

\subsubsection{Reduced ODE subsystem}\label{eq:reducedODEe}
In proceeding with a singular perturbation analysis, the first step consists of deriving the reduced ODE subsystem.
To this end, by setting $\tau_z = 0$ in~\eqref{eq:originalSystemsPDE}, the following nonlocal ODE is deduced for fixed $t$: 
\begin{subequations}\label{eq:ODeredudec0}
\begin{align}
\begin{split}
& \Lambda \dpd{z(\xi,t)}{\xi}  =\Sigma\bigl(\alpha(X(t),U(t))\bigr)\bigl[z(\xi,t) + (\mathscr{K}_3z)(t)\bigr] \\
& \qquad\qquad \qquad  +(\mathscr{K}_4z)(t)+h_2\Bigl(\alpha\bigl(X(t),U(t)\bigr)\Bigr), \\
& \qquad \qquad \qquad \qquad \qquad  (\xi,t) \in (0,1) \times (0,T),
\end{split}  \label{eq:originalSystemsPDEEqSF20} \\
& z(0,t) = 0, \quad t \in (0,T).\label{eq:originalSystemsBCEqSF20}
\end{align}
\end{subequations}
By Lemma~\ref{prop:pass}, for every $t \in (0,T)$,~\eqref{eq:ODeredudec0} admits a unique solution $z(\cdot,t) \in C^1([0,1];\mathbb{R}^2)$ satisfying the BC~\eqref{eq:originalSystemsBCEqSF20} in the form
\begin{align}\label{eq:zdeltaQA0}
z(\xi,t) = \hat{z}\Bigl(\xi,\alpha\bigl(X(t),U(t)\bigr)\Bigr), \quad \xi \in [0,1].
\end{align}
Substituting~\eqref{eq:zdeltaQA0} into~\eqref{eq:originalSystemsODE} yields
\begin{align}\label{eq:X_deltaRed0}
\begin{split}
\dot{\bar{X}}(t) & = A_1\bar{X}(t) + G_1\hat{F}_y\Bigl(\alpha\bigl(\bar{X}(t),\bar{U}(t)\bigr)\Bigr), \quad t \in (0,T),
\end{split}
\end{align}
where the bar notation has been introduced to indicate the ODE variables obtained for $\tau_z = 0$.

An equilibrium $(X^\star,U^\star) \in \mathbb{R}^4$ as in~\eqref{eq:equilibrium} is now considered, along with the variables $\mathbb{R}^2 \ni \bar{X}_\delta(t) \triangleq \bar{X}(t)-X^\star$ and $\mathbb{R}^2 \ni \bar{U}_\delta(t) \triangleq \bar{U}(t)-U^\star$. Consequently, performing an exact first-order Taylor's expansion under Assumption~\ref{ass:differ}, and specifying $\bar{U}_\delta(t) = 0$ provides
\begin{align}\label{eq:LienarXdeltraaa0}
\begin{split}
\dot{\bar{X}}_\delta(t) & = \tilde{A}_1(\alpha^\star)\bar{X}_\delta(t) + R\bigl(\bar{X}_\delta(t),\alpha^\star\bigr), \quad t \in (0,T),
\end{split}
\end{align}
where, for every $\alpha^\star \in \mathbb{R}^2$, $R(\cdot,\alpha^\star) \in C^0(\mathbb{R}^2;\mathbb{R}^2)$ satisfies
\begin{align}\label{eq:Lsossd0D0}
\lim_{\norm{Z}_2 \to 0}\dfrac{\norm{R(Z,\alpha^\star)}_2}{\norm{Z}_2} = 0.
\end{align}
Local exponential stability results for the reduced ODE subsystem~\eqref{eq:LienarXdeltraaa0}, which is equivalent to~\eqref{eq:X_deltaRed0}, are stated in Proposition~\ref{prop:localExpODEe} below.
\begin{proposition}\label{prop:localExpODEe}
Under Assumptions~\ref{ass:ODEf}-\ref{ass:differ}, consider the reduced ODE subsystem~\eqref{eq:LienarXdeltraaa0} along with the equilibrium $(X^\star,U^\star)$, and suppose that the matrix $\tilde{A}_1(\alpha^\star)\in \mathbf{M}_2(\mathbb{R})$ in~\eqref{eq:A1tilde} is Hurwitz. Then, there exists $r \in \mathbb{R}_{>0}$ such that, for all ICs $\bar{X}_0 \triangleq \bar{X}(0) \in \mathbb{R}^2$ with $\norm{\bar{X}_0-X^\star}_2 < r$, the unique solution $\bar{X} \in C^0(\mathbb{R}_{\geq 0};\mathbb{R}^2)$ to~\eqref{eq:X_deltaRed0} satisfies
\begin{align}
\norm{\bar{X}(t)-X^\star}_2 \leq \beta_1(\alpha^\star)\eu^{-\sigma_1(\alpha^\star) t}\norm{\bar{X}_0-X^\star}_2, \quad t\in \mathbb{R}_{\geq 0},
\end{align}
for some $\beta_1(\alpha^\star),\sigma_1(\alpha^\star) \in \mathbb{R}_{>0}$.
\end{proposition}
Based on the assertion of Proposition~\ref{prop:localExpODEe}, some interesting considerations are collected below.
\begin{remark}
For $\tilde{C}_1(\alpha_1^\star),\tilde{C}_2(\alpha_2^\star) >0$, the famous (local) \emph{understeer inequality} $\tilde{C}_1(\alpha_1^\star)l_1 < \tilde{C}_2(\alpha_2^\star)l_2$ encountered in the vehicle dynamics literature \parencite{Guiggiani} is sufficient to ensure that $\tilde{A}_1(\alpha^\star)$ is Hurwitz.
\end{remark} 
\begin{remark}
The reduced ODE subsystem, in both its original and linearized variants~\eqref{eq:X_deltaRed0} and~\eqref{eq:LienarXdeltraaa0}, respectively, coincides with the formulations typically encountered in the literature, obtained by disregarding the fast dynamics of the tires \parencite{Guiggiani}.
\end{remark}

\subsubsection{Boundary layer PDE subsystem}\label{eq:BoundaryLayere}
The second step requires deriving the boundary layer PDE subsystem.
In particular, performing the transformation $\mathbb{R}^2 \ni z_\delta(\xi,t) = [z_{\delta1}(\xi,t)\; z_{\delta2}(\xi,t)]^{\mathrm{T}} \triangleq z(\xi,t)-\hat{z}(\xi,\alpha(X(t),U(t)))$, and introducing the time-like variable $\mathbb{R}_{\geq 0} \ni s \triangleq t/\tau_z$ gives
\begin{subequations}\label{eq:ODeredudec2}
\begin{align}
\begin{split}
& \dpd{z_\delta(\xi,s)}{s}+ \Lambda \dpd{z_\delta(\xi,s)}{\xi} =\Sigma\bigl(\alpha(X(t),U(t))\bigr) \\
& \qquad \times \bigl[z_\delta(\xi,s) + (\mathscr{K}_3z_\delta)(s)\bigr]+(\mathscr{K}_4z_\delta)(s) \\
& \qquad - \tau_z\dpd{\hat{z}\bigl(\xi,\alpha(X(t),U(t))\bigr)}{t}, \quad (\xi,s) \in (0,1) \times (0,S),
\end{split}  \label{eq:originalSystemsPDEEqSF22} \\
& z_\delta(0,s) = 0, \quad s \in (0,S),\label{eq:originalSystemsBCEqSF22}
\end{align}
\end{subequations}
where the identity
\begin{align}
\begin{split}
& \Lambda\dpd{\hat{z}\bigl(\xi,\alpha(X(t),U(t))\bigr)}{\xi}  =\Sigma\bigl(\alpha(X(t),U(t))\bigr) \\
& \qquad \times \Bigl[\hat{z}\bigl(\xi,\alpha(X(t),U(t))\bigr) + (\mathscr{K}_3\hat{z})\bigl(\alpha(X(t),U(t))\bigr)\Bigr] \\
& \qquad  +(\mathscr{K}_4\hat{z})\bigl(\alpha(X(t),U(t))\bigr) \\
& \qquad +h_2\Bigl(\alpha\bigl(X(t),U(t)\bigr)\Bigr), \quad \xi \in [0,1],
\end{split}
\end{align}
has been used.
Therefore, by setting $\tau_z = 0$ in~\eqref{eq:ODeredudec2}, the dynamics of the boundary layer PDE subsystem may be deduced to obey
\begin{subequations}\label{eq:ODeredudec3}
\begin{align}
\begin{split}
& \dpd{\bar{z}_\delta(\xi,s)}{s}+ \Lambda \dpd{\bar{z}_\delta(\xi,s)}{\xi} =\Sigma\bigl(\alpha(X(t),U(t))\bigr) \\
& \qquad \qquad \times \bigl[\bar{z}_\delta(\xi,s) + (\mathscr{K}_3\bar{z}_\delta)(s)\bigr]   +(\mathscr{K}_4\bar{z}_\delta)(s), \\
& \qquad \qquad \qquad \qquad \quad (\xi,s) \in (0,1) \times (0,S),
\end{split}  \label{eq:originalSystemsPDEEqSF23} \\
& \bar{z}_\delta(0,s) = 0, \quad s \in (0,S),\label{eq:originalSystemsBCEqSF23}
\end{align}
\end{subequations}
where the bar notation has again been adopted to indicate the PDE variables corresponding to $\tau_z = 0$. Before enouncing the stability results for~\eqref{eq:ODeredudec3}, it is worth clarifying that, in~\eqref{eq:ODeredudec3}, the variable $t$ should be regarded as a parameter, and consequently $X(t)$ and $U(t)$ as frozen at a certain time.
\begin{lemma}\label{lemma:sdjda}
Suppose that Assumptions~\ref{ass:ODEf} and~\ref{ass:ODEf2} hold. Then, for all ICs $\bar{z}_{\delta,0} \triangleq \bar{z}_\delta(\cdot,0) \in L^2((0,1);\mathbb{R}^2)$ and $(X(t),U(t)) \in \mathbb{R}^4$, the boundary layer PDE subsystem~\eqref{eq:ODeredudec3} admits a unique mild solution 
$\bar{z}_\delta \in C^0(\mathbb{R}_{\geq 0};L^2((0,1);\mathbb{R}^2))$ satisfying
\begin{align}\label{eq:bound1}
\norm{\bar{z}_\delta(\cdot,s)}_{L^2((0,1);\mathbb{R}^2)} \leq \beta_2\eu^{-\sigma_2 s}\norm{\bar{z}_{\delta,0}(\cdot)}_{L^2((0,1);\mathbb{R}^2)}, \quad s \in \mathbb{R}_{\geq 0},
\end{align}
for some $\beta_2,\sigma_2 \in \mathbb{R}_{>0}$.
\begin{proof}
According to~\textcite[Theorem~2.1 and Corollary~2.1]{FEM}, the operator $(\mathscr{A}_\Sigma(y), \mathscr{D}(\mathscr{A}_\Sigma(y)))$ defined in~\eqref{eq:Ay} generates a $C_0$-semigroup on $L^2((0,1);\mathbb{R}^2)$ for every $y \in \mathbb{R}^2$. Hence, setting $y = \alpha(X(t),U(t))$, the existence and uniqueness of mild solutions $\bar{z}_\delta \in C^0(\mathbb{R}_{\geq 0};L^2((0,1);\mathbb{R}^2))$ follow by standard arguments for all $\bar{z}_{\delta,0} \in L^2((0,1);\mathbb{R}^2)$ and $(X(t),U(t)) \in \mathbb{R}^4$. On the other hand, the bound~\eqref{eq:bound1} may be obtained by defining the Lyapunov function $\bar{V}(\bar{z}_\delta(\cdot,s)) \triangleq \frac{1}{2}\langle \bar{z}_\delta(\cdot,s),\mathscr{Q}\bar{z}_\delta(\cdot,s)\rangle_{L^2((0,1);\mathbb{R}^2)}$, which, together with Assumptions~\ref{ass:ODEf} and~\ref{ass:ODEf2}, implies
\begin{align}
\dot{\bar{V}}(s) \leq -\omega\norm{\bar{z}_\delta(\cdot,s)}_{L^2((0,1);\mathbb{R}^2)}^2, \quad s \in \mathbb{R}_{\geq 0}.
\end{align}
Since $\bar{V}(\bar{z}_\delta(\cdot,s))$ is equivalent to a squared norm on $L^2((0,1);\mathbb{R}^2)$,~\eqref{eq:bound1} follows with $\beta_2,\sigma_2 \in \mathbb{R}_{>0}$ uniform in the frozen parameters $(X(t),U(t))$. 
\end{proof}
\end{lemma}

The assertions of Proposition~\ref{prop:localExpODEe} and Lemma~\ref{lemma:sdjda} are formally valid only for $\tau_z = 0$. The next Section~\ref{sect:StateFedde} generalizes the singular perturbation analysis by considering sufficiently small $\tau_z$.

\subsection{Stability analysis\\ for sufficiently small $\tau_z$}\label{sect:StateFedde}
Based on the preliminary results of Section~\ref{sect:StateFedd0}, the local stability of the system~\eqref{eq:originalSystems} is studied for sufficiently small $\tau_z$. To this end, the ODE-PDE interconnection~\eqref{eq:originalSystems} may be recast as
\begin{subequations}\label{eq:LienarXdeltraaaBXCDe}
\begin{align}\label{eq:LienarXdeltraaaBe}
\begin{split}
& \dot{X}_\delta(t) = \tilde{A}_1(\alpha^\star)X_\delta(t) + R\bigl(X_\delta(t),\alpha^\star\bigr) +G_1(\mathscr{K}_1z_\delta)(t) \\
& \qquad \qquad +G_1\Sigma\bigl(\alpha(X(t),U^\star)\bigr)(\mathscr{K}_2z_\delta)(t), \quad t \in (0,T),
\end{split} \\
\begin{split}
& \dpd{z_\delta(\xi,t)}{t}+ \dfrac{\Lambda}{\tau_z} \dpd{z_\delta(\xi,t)}{\xi} =\dfrac{1}{\tau_z}\Sigma\bigl(\alpha(X(t),U^\star)\bigr) \\
& \qquad\qquad \times \bigl[z_\delta(\xi,t) + (\mathscr{K}_3z_\delta)(t)\bigr]+\dfrac{1}{\tau_z}(\mathscr{K}_4z_\delta)(t) \\
& \qquad\qquad- \dpd{\hat{z}\bigl(\xi,\alpha(X(t),U^\star)\bigr)}{\alpha}A_2 \\
& \qquad\qquad \times\bigl[\tilde{A}_1(\alpha^\star)X_\delta(t) + R\bigl(X_\delta(t),\alpha^\star\bigr)\bigr] \\
& \qquad\qquad - \dpd{\hat{z}\bigl(\xi,\alpha(X(t),U^\star)\bigr)}{\alpha}A_2G_1 \\
& \qquad\qquad \times \Bigl[(\mathscr{K}_1z_\delta)(t) +\Sigma\bigl(\alpha(X(t),U^\star)\bigr)(\mathscr{K}_2z_\delta)(t)\Bigr], \\
& \qquad\qquad\qquad \qquad \quad (\xi,t) \in (0,1) \times (0,T),
\end{split}  \label{eq:originalSystemsPDEEqSF22e} \\
& z_\delta(0,t) = 0, \quad t \in (0,T),\label{eq:originalSystemsBCEqSF22e}
\end{align}
\end{subequations}
The following Lyapunov function candidate is considered:
\begin{align}\label{eq:V1e}
V_1\bigl(X_\delta(t)\bigr)\triangleq \dfrac{1}{2}X_\delta^{\mathrm{T}}(t)Q(\alpha^\star)X_\delta(t),
\end{align}
where $\mathbf{Sym}_2(\mathbb{R}) \ni Q(\alpha^\star) \succ 0$ is chosen such that $\tilde{A}_1^{\mathrm{T}}(\alpha^\star)Q(\alpha^\star) + Q(\alpha^\star)\tilde{A}_1(\alpha^\star) = -2qI_2$ for some $q\in \mathbb{R}_{> 0}$.
Differentiating~\eqref{eq:V1e} along the dynamics~\eqref{eq:LienarXdeltraaaBe} yields
\begin{align}\label{eq:V1dot1e}
\begin{split}
\dot{V}_1(t) & \leq - q\norm{X_\delta(t)}_2^2+ X_\delta^{\mathrm{T}}(t)Q(\alpha^\star)G_1 \\
& \quad \times \Bigl[(\mathscr{K}_1z_\delta)(t) +\Sigma\bigl(\alpha(X(t),U^\star)\bigr)(\mathscr{K}_2z_\delta)(t)\Bigr] \\
& \quad + \norm{X_\delta(t)}_2\norm{Q(\alpha^\star)}\norm{R\bigl(X_\delta(t),\alpha^\star\bigr)}_2, \quad t \in (0,T).
\end{split}
\end{align}
By Assumption~\ref{ass:ODEf2}, there exists $b_\Sigma(\alpha^\star) \in \mathbb{R}_{\geq 0}$ such that
\begin{align}\label{eq:boundSigmab}
\begin{split}
\norm{\Sigma\bigl(\alpha(X(t),U^\star)\bigr)} & \leq L_\Sigma\norm{A_2} \norm{X_\delta(t)}_2 + b_\Sigma(\alpha^\star),
\end{split}
\end{align}
which gives
\begin{align}\label{eq:V1dot12e}
\begin{split}
\dot{V}_1(t) & \leq - q\norm{X_\delta(t)}_2^2 + \eta_1(\alpha^\star)\norm{X_\delta(t)}_2\norm{z_\delta(\cdot,t)}_{L^2((0,1);\mathbb{R}^2)}\\
& \quad +\eta_2 \norm{X_\delta(t)}_2\biggl(\norm{X_\delta(t)}_2\norm{z_\delta(\cdot,t)}_{L^2((0,1);\mathbb{R}^2)}  \\
& \quad +  \norm{R\bigl(X_\delta(t),\alpha^\star\bigr)}_2\biggr), \quad t \in (0,T).
\end{split}
\end{align}
with
\begin{subequations}
\begin{align}
\eta_1(\alpha^\star) & \triangleq \norm{Q(\alpha^\star)G_1}\Bigl(\norm{K_1(\cdot)}_\infty +b_\Sigma(\alpha^\star)\norm{K_2(\cdot)}_\infty\Bigr), \\
\eta_2 & \triangleq \max\Bigl\{\norm{Q(\alpha^\star)G_1}L_\Sigma\norm{A_2}\norm{K_2(\cdot)}_\infty, \norm{Q(\alpha^\star)} \Bigr\}.
\end{align}
\end{subequations}
A second Lyapunov function candidate is then defined as
\begin{align}\label{eq:V2e}
V_2\bigl(z_\delta(\cdot,t)\bigr) \triangleq \dfrac{1}{2}\int_0^1 z_\delta^{\mathrm{T}}(\xi,t)\mathscr{Q}(\xi)z_\delta(\xi,t) \dif \xi.
\end{align}
Differentiating~\eqref{eq:V2e} along the dynamics~\eqref{eq:originalSystemsPDEEqSF22e} and imposing the BC~\eqref{eq:originalSystemsBCEqSF22e} yields
\begin{align}\label{eq:V2dot1e}
\begin{split}
\dot{V}_2(t) & \leq -\dfrac{\omega}{\tau_z}\norm{z_\delta(\cdot,t)}_{L^2((0,1);\mathbb{R}^2)}^2 + \eta_3\norm{z_\delta(\cdot,t)}_{L^2((0,1);\mathbb{R}^2)} \\
& \quad \times \biggl(\norm{\tilde{A}_1(\alpha^\star)}\norm{X_\delta(t)}_2+\norm{R\bigl(X_\delta(t),\alpha^\star\bigr)}_2 \\
& \quad + \norm{G_1}\norm{K_2(\cdot)}_\infty\norm{\Sigma\bigl(\alpha(X(t),U^\star)\bigr)} \norm{z_\delta(\cdot,t)}_{L^2((0,1);\mathbb{R}^2)}  \\
& \quad +\norm{G_1K_1(\cdot)}_\infty\norm{z_\delta(\cdot,t)}_{L^2((0,1);\mathbb{R}^2)} \biggr),  \quad t \in (0,T),
\end{split}
\end{align}
with
\begin{align}
\eta_3 \triangleq M_{\hat{z}}\norm{\mathscr{Q}(\cdot)}_\infty\norm{A_2}.
\end{align}
Utilizing again~\eqref{eq:boundSigmab} provides
\begin{align}
\begin{split}
\dot{V}_2(t) & \leq -\dfrac{\omega}{\tau_z}\biggl(1-\tau_z\dfrac{\eta_4(\alpha^\star)}{\omega} \biggr)\norm{z_\delta(\cdot,t)}_{L^2((0,1);\mathbb{R}^2)}^2 \\
& \quad + \eta_5\norm{X_\delta(t)}_2\norm{z_\delta(\cdot,t)}_{L^2((0,1);\mathbb{R}^2)} \\
& \quad +\eta_6\norm{z_\delta(\cdot,t)}_{L^2((0,1);\mathbb{R}^2)} \biggl(\norm{R\bigl(X_\delta(t),\alpha^\star\bigr)}_2 \\
& \quad +\norm{X_\delta(t)}_2 \norm{z_\delta(\cdot,t)}_{L^2((0,1);\mathbb{R}^2)} \biggr), \quad t \in (0,T).
\end{split}
\end{align}
with
\begin{subequations}
\begin{align}
\eta_4(\alpha^\star) & \triangleq \eta_3\Bigl( \norm{G_1K_1(\cdot)}_\infty+ \norm{G_1}\norm{K_2(\cdot)}_\infty b_\Sigma(\alpha^\star)\Bigr), \\
\eta_5(\alpha^\star) & \triangleq \eta_3\norm{\tilde{A}_1(\alpha^\star)},\\
\eta_6 & \triangleq \eta_3\max\Bigl\{1,L_\Sigma\norm{G_1}\norm{K_2(\cdot)}_\infty\norm{A_2}\Bigr\}.
\end{align}
\end{subequations}
The final Lyapunov function is consequently assembled as
\begin{align}\label{eq:VLyapCOmpl}
V\bigl(X_\delta(t),z_\delta(\cdot,t)\bigr) \triangleq V_1\bigl(X_\delta(t)\bigr) + V_2\bigl(z_\delta(\cdot,t)\bigr).
\end{align}
Accordingly, a straightforward application of Cauchy-Schwarz and the generalized Young's inequality for products provides
\begin{align}\label{eq:Vdot1e}
\begin{split}
\dot{V}(t) & \leq - \dfrac{q}{2}\norm{X_\delta(t)}_2^2 \\
& \quad  -\dfrac{\omega}{\tau_z}\biggl(1-\tau_z\dfrac{\eta_7(\alpha^\star)}{\omega}\biggr)\norm{z_\delta(\cdot,t)}_{L^2((0,1);\mathbb{R}^2)}^2 \\
& \quad + \eta_8\biggl( \norm{X_\delta(t)}_2 + \norm{z_\delta(\cdot,t)}_{L^2((0,1);\mathbb{R}^2)}\biggr)\\
& \quad \times \biggl(\norm{R\bigl(X_\delta(t),\alpha^\star\bigr)}_2  \\
& \quad +\norm{X_\delta(t)}_2 \norm{z_\delta(\cdot,t)}_{L^2((0,1);\mathbb{R}^2)} \biggr),  \quad t \in (0,T),
\end{split}
\end{align}
where
\begin{subequations}
\begin{align}
\eta_7(\alpha^\star) & \triangleq\eta_4(\alpha^\star) + \dfrac{\bigl(\eta_1(\alpha^\star)+\eta_5(\alpha^\star)\bigr)^2}{2q}, \\
\eta_8& \triangleq \max\{\eta_2,\eta_6\}. 
\end{align}
\end{subequations}
Therefore, it may be inferred that there exists $\tau_z^*(\alpha^\star) \in \mathbb{R}_{>0}$ such that, for every $\tau_z \in (0,\tau_z^*(\alpha^\star))$,
\begin{align}\label{eq:Vdot3e}
\begin{split}
\dot{V}(t) & \leq -\gamma_0(\alpha^\star)V(t) \\
& \quad + \eta_9(\alpha^\star)V^{\frac{1}{2}}(t)\biggl(V(t) + \norm{R\bigl(X_\delta(t),\alpha^\star\bigr)}_2\biggr), \\ & \qquad \qquad \qquad \qquad \qquad \qquad \qquad \quad t \in (0,T),
\end{split}
\end{align}
for some $\gamma_0(\alpha^\star), \eta_9(\alpha^\star) \in \mathbb{R}_{>0}$. Moreover,~\eqref{eq:Lsossd0D0} implies that, for every $\gamma_1\in \mathbb{R}_{>0}$, there exists $r_\tau^*(\alpha^\star)$ such that $\norm{R(X_\delta(t),\alpha^\star)}_2 \leq \gamma_1V^{\frac{1}{2}}(t)$ for all $V(t) < r_\tau^*(\alpha^\star)$. Consider now $r_\tau\in (0,r_\tau^*(\alpha^\star))$. For $V(t) < r_\tau$, 
\begin{align}\label{eq:Vdot4}
\begin{split}
\dot{V}(t) & \leq -\gamma_0(\alpha^\star)V(t)  + \eta_9(\alpha^\star)\bigl(\sqrt{r_\tau} +\gamma_1\bigr)V(t), \quad t \in (0,T).
\end{split}
\end{align}
Hence, choosing $r_\tau$ such that $\eta_9(\alpha^\star)(\sqrt{r_\tau} +\gamma_1)< \gamma_0(\alpha^\star)$ guarantees the existence of $\gamma(\alpha^\star) \in \mathbb{R}_{>0}$ such that 
\begin{align}\label{eq:Vfinal}
\dot{V}(t) \leq - \gamma(\alpha^\star)V(t), \quad t \in (0,T),
\end{align}
for all $V(t) < r_\tau$. 
The next Theorem~\ref{thm:main1} asserts the main result of the paper.
\begin{theorem}\label{thm:main1}
Under Assumptions~\ref{ass:ODEf}-\ref{ass:differ}, consider the ODE-PDE interconnection~\eqref{eq:originalSystems} along with the equilibrium $(X^\star,U^\star)$, and the input $U(t) = U^\star$, and suppose that the matrix $\tilde{A}_1(\alpha^\star) \in \mathbf{M}_2(\mathbb{R})$ in~\eqref{eq:A1tilde} is Hurwitz. Then, there exist $\tau_z^*(\alpha^\star), r_\tau \in \mathbb{R}_{>0}$ such that, for all $\tau_z \in (0,\tau_z^*(\alpha^\star))$ and ICs $(X_0,z_0(\cdot)) \in \altmathcal{X}$ verifying $V(0) < r_\tau$, with $V(X_\delta(t),z_\delta(\cdot,t))$ as in~\eqref{eq:VLyapCOmpl}, the ODE-PDE system~\eqref{eq:originalSystems} admits a unique mild solution $(X,z) \in C^0(\mathbb{R}_{\geq 0};\altmathcal{X})$ satisfying
\begin{align}\label{eq:boundFinalV0}
\begin{split}
& \norm{(X(t)-X^\star, z_\delta(\cdot,t))}_{\altmathcal{X}} \\
& \quad \leq \beta(\alpha^\star)\eu^{-\sigma(\alpha^\star) t}\norm{(X_0-X^\star, z_{\delta,0}(\cdot))}_{\altmathcal{X}}, \quad t \in \mathbb{R}_{\geq 0},
\end{split}
\end{align}
for some $\beta(\alpha^\star),\sigma(\alpha^\star) \in \mathbb{R}_{>0}$.
\begin{proof}
It follows from standard semigroup arguments for semilinear problems (see, e.g., \textcite[Theorem~11.1.5]{Zwart} or \textcite[Theorem~6.1.4]{Pazy}) that the closed-loop ODE-PDE interconnection~\eqref{eq:LienarXdeltraaaBXCDe} admits a unique local mild solution $(X_\delta,z_\delta) \in C^0([0,t\ped{max});\altmathcal{X})$ for all ICs $(X_{\delta,0},z_{\delta,0})\triangleq (X_\delta(0),z_\delta(\cdot,0)) \in \altmathcal{X}$. Consequently, from the transformations $X(t) \triangleq X_\delta(t) +X^\star$ and $z(\xi,t) \triangleq z_\delta(\xi,t) + \hat{z}(\xi,\alpha(X(t),U(t)))$, with $U(t) = U^\star$ and $\hat{z}(\cdot,\alpha(X(t),U(t))) \in \mathscr{D}(\mathscr{A})$, it may be concluded the original ODE-PDE interconnection~\eqref{eq:originalSystems} also admits a unique mild solution $(X,z) \in C^0([0,t\ped{max});\altmathcal{X})$ for all ICs $(X_0,z_0)\in \altmathcal{X}$.
Moreover, according to~\textcite[Theorem~11.1.5]{Zwart}, to prove global well-posedness, it is sufficient to show that $\norm{(X(t),z(\cdot,t))}_{\altmathcal{X}} < \infty$ for all $t \in \mathbb{R}_{\geq 0}$, which is implied by inequality~\eqref{eq:Vfinal} under Assumption~\ref{ass:differ}. In particular, for sufficiently regular solutions and for $V(0) < r_\tau$,~\eqref{eq:Vfinal} gives $V(t) \leq \eu^{-\gamma(\alpha^\star)}V(0) < r_\tau$ for all $t \in \mathbb{R}_{\geq 0}$; in conjunction with the fact that the Lyapunov function $V(X_\delta(t),z_\delta(\cdot,t))$ is equivalent to the squared norm $\norm{(X_\delta(t),z_\delta(\cdot,t))}_{\altmathcal{X}}^2$ on $\altmathcal{X}$, this yields
\begin{align}\label{eq:XNormksks}
\begin{split}
& \norm{(X_\delta(t),z_\delta(\cdot,t))}_{\altmathcal{X}} \\
& \quad \leq \beta(\alpha^\star)\eu^{-\sigma(\alpha^\star) t}\norm{(X_{\delta,0},z_{\delta,0}(\cdot))}_{\altmathcal{X}}, \quad t \in \mathbb{R}_{\geq 0},
\end{split}
\end{align}
which proves~\eqref{eq:boundFinalV0}.
The result may then be extended to mild solutions using standard density arguments.
\end{proof}
\end{theorem}
Before moving to the design of a stabilizing controller, some conclusive remarks are collected below. 
\begin{remark}\label{remark:tveh}
Theorem~\ref{thm:main1} states that, for sufficiently small $\tau_z$ (and thus high $v_x$), the stability of the semilinear single-track models may be studied using standard finite-dimensional approaches.
In this context, the expression “sufficiently small $\tau_z$” -- and thus the degree of time scale separation -- may be interpreted relative to the (local) characteristic time scale $\tau_X(\alpha^\star) \in \mathbb{R}_{>0}$ of the reduced vehicle dynamics governed by $\tilde A_1(\alpha^\star)$, calculated as
\begin{align}\label{eq:tau_xOpen}
\tau_X(\alpha^\star) = \frac{1}{\abs{\max_{i \in \{1,2\}}\RE \lambda_i(\tilde{A}_1(\alpha^\star))}},
\end{align} 
where $\lambda_1(\tilde{A}_1(\alpha^\star))$ and $\lambda_2(\tilde{A}_1(\alpha^\star))$ denote the eigenvalues of $\tilde{A}_1(\alpha^\star)$. Equation~\eqref{eq:tau_xOpen} is valid provided that the spectral abscissa of $\tilde{A}_1(\alpha^\star)$ is nonzero. When $\tilde{A}_1(\alpha^\star)$ is Hurwitz, $\tau_X(\alpha^\star)$ represents the dominant decay time; when it is unstable, it represents the dominant growth time. The dimensionless ratio $\tau_z/\tau_X(\alpha^\star)$ may thus be regarded as a local indicator of the degree of time-scale separation, but not as a sufficient stability criterion for the full ODE-PDE interconnection.
\end{remark}

A counterexample to the applicability of classical finite-dimensional tools is provided graphically in Figure~\ref{figureForcePostdocrg}, where a local stability chart is illustrated for a semilinear single-track vehicle model~\eqref{eq:originalSystems} with distributed tires, using typical parameter values \parencite{SemilinearV}. The chart, produced using spectral methods, reveals the existence of unstable islands in understeer vehicles ($\tilde{C}_1(0)l_1 < \tilde{C}_2(0)l_2$) for sufficiently low values of the longitudinal speed $v_x$. As discussed in \textcite{Takacs5,BicyclePDE,SemilinearV}, the unstable regions in Figure~\ref{figureForcePostdocrg} are associated with micro-shimmy oscillations and are not predicted by the reduced order model~\eqref{eq:X_deltaRed0}.
\begin{figure*}
\centering
\includegraphics[width=0.7\linewidth]{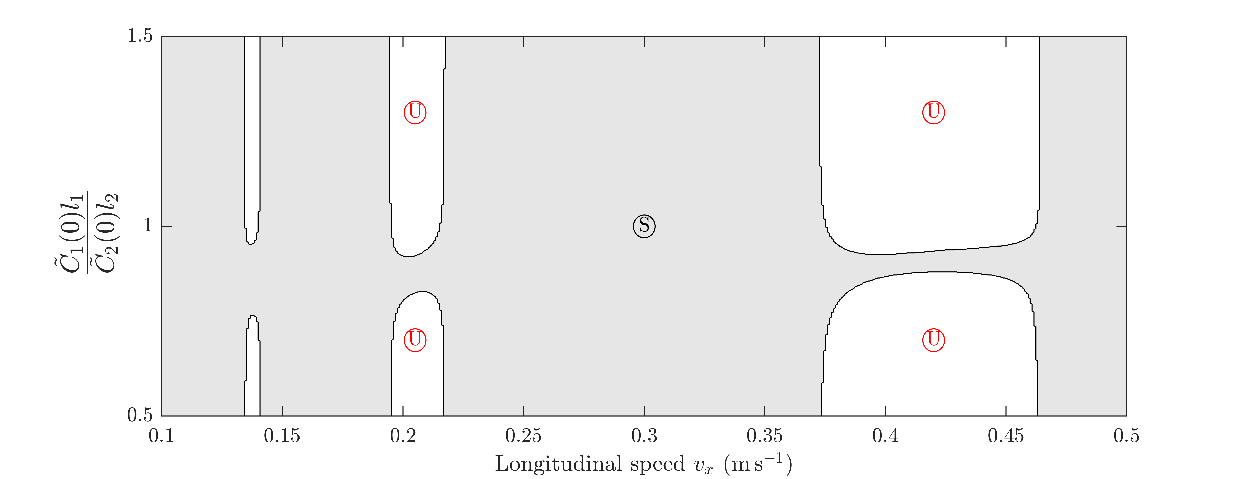} 
\caption{Local stability chart for a semilinear single-track vehicle model~\eqref{eq:originalSystems} with constant pressure distribution and flexible tire carcass (Parametrization~\ref{param:2}), linearized around the zero equilibrium $(X^\star, z^\star(\xi), U^\star) = (0,0,0)$, for different values of the understeer index $\frac{\tilde{C}_{1}(0)l_1}{\tilde{C}_{2}(0)l_2}$ and longitudinal speed $v_x$. The unstable regions (in white) correspond to combinations of parameters associated with micro-shimmy oscillations that are not detected by reduced order representations. Model parameters as in \textcite{SemilinearV}.}
\label{figureForcePostdocrg}
\end{figure*}


\section{Stabilization}\label{sect:Controllllll}
Single-track vehicle models -- especially \emph{oversteer} ones -- might be unstable.
Therefore, the present section is dedicated to the synthesis of a controller that (locally) exponentially stabilizes~\eqref{eq:originalSystems} around a desired equilibrium $X^\star \in \mathbb{R}^2$ as in~\eqref{eq:equilibrium}, corresponding to a stationary input $U(t) = U^\star$.

\subsection{Analysis via singular perturbation\\ theory}\label{sect:StateFedd}
Repeating analogous calculations as in Section~\ref{eq:reducedODEe} and specifying the input as $\bar{U}(t) = U^\star + \bar{U}_\delta(t)$, with
\begin{align}\label{eq:Udeltasss}
\bar{U}_\delta(t) = F(\alpha^\star)\bar{X}_\delta(t),
\end{align}
where $F(\alpha^\star)\in \mathbf{M}_2(\mathbb{R})$ is a matrix to be appropriately selected, produces, after an exact first-order Taylor's expansion,
\begin{align}\label{eq:LienarXdeltraaa}
\begin{split}
\dot{\bar{X}}_\delta(t) & = \tilde{A}_1^*(\alpha^\star)\bar{X}_\delta(t) + R^*\bigl(\bar{X}_\delta(t),\alpha^\star\bigr), \quad t \in (0,T),
\end{split}
\end{align}
where the matrix $\tilde{A}_1^*(\alpha^\star) \in \mathbf{M}_2(\mathbb{R})$ is defined as
\begin{align}\label{eq:A1TildeStar}
\tilde{A}_1^*(\alpha^\star) & \triangleq \tilde{A}_1(\alpha^\star) + \tilde{G}_1(\alpha^\star)F(\alpha^\star),
\end{align}
and, for every $\alpha^\star \in \mathbb{R}^2$, $R^*(\cdot,\alpha^\star) \in C^0(\mathbb{R}^2;\mathbb{R}^2)$ satisfies
\begin{align}\label{eq:Lsossd0D}
\lim_{\norm{Z}_2 \to 0}\dfrac{\norm{R^*(Z,\alpha^\star)}_2}{\norm{Z}_2} = 0.
\end{align}
Local exponential stability results for the closed-loop reduced ODE subsystem~\eqref{eq:LienarXdeltraaa} are asserted by Proposition~\ref{prop:localExpODE} below.
\begin{proposition}\label{prop:localExpODE}
Under Assumptions~\ref{ass:ODEf}-~\ref{ass:stabilizability}, consider the reduced ODE subsystem~\eqref{eq:LienarXdeltraaa} along with the equilibrium $(X^\star,U^\star)$, and the control law $\bar{U}(t) = U^\star+\bar{U}_\delta(t)$, with $\bar{U}_\delta(t)$ as in~\eqref{eq:Udeltasss} and the gain $F(\alpha^\star) \in \mathbf{M}_2(\mathbb{R})$ designed such that the matrix $\tilde{A}_1^*(\alpha^\star) \in \mathbf{M}_2(\mathbb{R})$ in~\eqref{eq:A1TildeStar} is Hurwitz. Then, there exists $r^* \in \mathbb{R}_{>0}$ such that, for all ICs $\bar{X}_0 \triangleq \bar{X}(0) \in \mathbb{R}^2$ with $\norm{\bar{X}_0-X^\star}_2 < r^*$, the unique solution $\bar{X} \in C^0(\mathbb{R}_{\geq 0};\mathbb{R}^2)$ to~\eqref{eq:LienarXdeltraaa} satisfies
\begin{align}
\norm{\bar{X}(t)-X^\star}_2 \leq \beta_1^*(\alpha^\star)\eu^{-\sigma_1^*(\alpha^\star) t}\norm{\bar{X}_0-X^\star}_2, \quad t\in \mathbb{R}_{\geq 0},
\end{align}
for some $\beta_1^*(\alpha^\star),\sigma_1^*(\alpha^\star) \in \mathbb{R}_{>0}$.
\end{proposition}
The boundary layer PDE subsystem is still given by~\eqref{eq:ODeredudec3}.

\begin{remark}
Since the pair $(\tilde{A}_1(\alpha^\star),\tilde{G}_1(\alpha^\star))$ is stabilizable by Assumption~\ref{ass:stabilizability}, the gain matrix $F(\alpha^\star)$ in~\eqref{eq:Udeltasss} may be selected using any standard finite-dimensional design technique, e.g., pole placement. In this way, the local exponential decay rate of the reduced subsystem~\eqref{eq:LienarXdeltraaa}, $\sigma_1^*(\alpha^\star) \in \mathbb{R}_{>0}$, can be increased by shifting the eigenvalues of $\tilde{A}_1^*(\alpha^\star)$ further into the open left half-plane. In particular, for all-wheel steering vehicles, the matrix $\tilde{G}_1(\alpha^\star)$ is typically invertible, so that $F(\alpha^\star)$ may be chosen as $F(\alpha^\star) = -\tilde{G}_1^{-1}(\alpha^\star)[\rho I_2 + \tilde{A}_1(\alpha^\star)]$ for some $\rho \in \mathbb{R}_{>0}$. In this context, it is worth emphasizing that the rate $\sigma_1^*(\alpha^\star)$ is assignable only for the reduced ODE subsystem~\eqref{eq:LienarXdeltraaa}, and not for the whole interconnection. In particular, assigning a larger decay rate to the reduced ODE subsystem does not necessarily imply a faster or better-damped response of the full ODE-PDE interconnection.
\end{remark}

\subsection{Stability analysis\\ for sufficiently small $\tau_z$}\label{sect:SFepsilon}
Inspired by the analysis conducted in the previous Section~\ref{sect:StateFedd}, the control law is specified as $U(t) = U^\star + U_\delta(t)$, with $U_\delta(t) \in \mathbb{R}^2$ given by
\begin{align}\label{eq:controlLwX}
U_\delta(t) = F(\alpha^\star)X_\delta(t).
\end{align}
With the above control law~\eqref{eq:controlLwX}, the ODE-PDE interconnection~\eqref{eq:originalSystems} may be recast as
\begin{subequations}\label{eq:LienarXdeltraaaBXCD}
\begin{align}\label{eq:LienarXdeltraaaB}
\begin{split}
& \dot{X}_\delta(t) = \tilde{A}_1^*(\alpha^\star)X_\delta(t) + R^*\bigl(X_\delta(t),\alpha^\star\bigr) +G_1(\mathscr{K}_1z_\delta)(t) \\
& \qquad \qquad +G_1\Sigma\bigl(\alpha(X(t),U(t))\bigr)(\mathscr{K}_2z_\delta)(t), \quad t \in (0,T),
\end{split} \\
\begin{split}
& \dpd{z_\delta(\xi,t)}{t}+ \dfrac{\Lambda}{\tau_z} \dpd{z_\delta(\xi,t)}{\xi} =\dfrac{1}{\tau_z}\Sigma\bigl(\alpha(X(t),U(t))\bigr) \\
& \qquad \times \bigl[z_\delta(\xi,t) + (\mathscr{K}_3z_\delta)(t)\bigr]+\dfrac{1}{\tau_z}(\mathscr{K}_4z_\delta)(t) \\
& \qquad - \dpd{\hat{z}\bigl(\xi,\alpha(X(t),U(t))\bigr)}{\alpha}\bigl(A_2 + G_2F(\alpha^\star)\bigr) \\
& \qquad \times \bigl[\tilde{A}_1^*(\alpha^\star)X_\delta(t) + R^*\bigl(X_\delta(t),\alpha^\star\bigr) +G_1(\mathscr{K}_1z_\delta)(t)\bigr] \\
& \qquad - \dpd{\hat{z}\bigl(\xi,\alpha(X(t),U(t))\bigr)}{\alpha}\bigl(A_2 + G_2F(\alpha^\star)\bigr) \\
& \qquad \times G_1\Sigma\bigl(\alpha(X(t),U(t))\bigr)(\mathscr{K}_2z_\delta)(t), \\
& \qquad \qquad \quad (\xi,t) \in (0,1) \times (0,T),
\end{split}  \label{eq:originalSystemsPDEEqSF22} \\
& z_\delta(0,t) = 0, \quad t \in (0,T),\label{eq:originalSystemsBCEqSF22}
\end{align}
\end{subequations}
Theorem~\ref{thm:main2} below represents the stabilization counterpart of Theorem~\ref{thm:main1}.
\begin{theorem}\label{thm:main2}
Under Assumptions~\ref{ass:ODEf}-\ref{ass:stabilizability}, consider the ODE-PDE interconnection~\eqref{eq:originalSystems} along with the equilibrium $(X^\star,U^\star)$, and the control law $U(t) = U^\star+U_\delta(t)$, with $U_\delta(t)$ as in~\eqref{eq:controlLwX} and the gain $F(\alpha^\star) \in \mathbf{M}_2(\mathbb{R})$ designed such that the matrix $\tilde{A}_1^*(\alpha^\star) \in \mathbf{M}_2(\mathbb{R})$ in~\eqref{eq:A1TildeStar} is Hurwitz. Then, there exist $\tau_z^{**}(\alpha^\star), r_\tau^*(\alpha^\star) \in \mathbb{R}_{>0}$ such that, for all $\tau_z \in (0,\tau_z^{**}(\alpha^\star))$ and ICs $(X_0,z_0(\cdot)) \in \altmathcal{X}$ verifying $\norm{(X_0-X^\star,z_{\delta,0}(\cdot))}_{\altmathcal{X}} < r_\tau^*(\alpha^\star)$, the ODE-PDE system~\eqref{eq:originalSystems} admits a unique mild solution $(X,z) \in C^0(\mathbb{R}_{\geq 0};\altmathcal{X})$ satisfying
\begin{align}\label{eq:boundFinalV}
\begin{split}
& \norm{(X(t)-X^\star, z_\delta(\cdot,t))}_{\altmathcal{X}} \\
& \quad \leq \beta^*(\alpha^\star)\eu^{-\sigma^*(\alpha^\star) t}\norm{(X_0-X^\star, z_{\delta,0}(\cdot))}_{\altmathcal{X}}, \quad t \in \mathbb{R}_{\geq 0},
\end{split}
\end{align}
for some $\beta^*(\alpha^\star),\sigma^*(\alpha^\star) \in \mathbb{R}_{>0}$.
\begin{proof}
The proof is similar to that of Theorem~\ref{thm:main1}, and thus omitted.
\end{proof}
\end{theorem}

\begin{remark}\label{remark:tveh2}
Similar to the open-loop case, the degree of time-scale separation may be studied by considering the ratio $\tau_z/\tau_X^*(\alpha^\star)$, with the time constant $\tau_X^*(\alpha^\star) \in \mathbb{R}_{>0}$ calculated as
\begin{align}\label{eq:tau_xClosed}
\tau_X^*(\alpha^\star) = \frac{1}{\abs{\max_{i \in \{1,2\}}\RE \lambda_i(\tilde{A}_1^*(\alpha^\star))}}.
\end{align} 
\end{remark}

\section{Simulation results}\label{sect:numerical}
The numerical values for the model parameters of the example discussed below are listed in Table~\ref{tab:parameters2}. With the given combination of parameters, Assumptions~\ref{ass:ODEf}-\ref{ass:stabilizability} are all fulfilled. Moreover, the considered vehicle is oversteer, and hence inherently unstable for values of the longitudinal speed beyond a critical value. 

In this context, the following numerical results refer to simulations conducted in MATLAB/Simulink\textsuperscript{\textregistered} environment. The semilinear PDE subsystem was solved numerically using a finite difference scheme with a discretization step of $0.02$, and then integrated with \texttt{ode4} with a fixed time step of $10^{-6}$ s. The ICs for the actual system were set to $X_0 = [0.03\; -0.25\, \textnormal{rad}\,\textnormal{s}^{-1}]^{\mathrm{T}}$, and $z_0(\xi) = [0.027\; 0.033]^{\mathrm{T}}$ (corresponding to $\norm{z_0(\cdot)}_{L^2((0,1);\mathbb{R}^2)} = 0.043$), and an input delay of $\delta_U = 0.02$ s was introduced. 
\begin{table}[ht]
\centering
\resizebox{\columnwidth}{!}{%
\begin{tabularx}{\columnwidth}{|c|X|c|c|}
\hline
\textbf{Parameter} & \textbf{Description} & \textbf{Unit} & \textbf{Value} \\
\hline 
$v_x$ & Longitudinal speed & $\textnormal{m}\,\textnormal{s}^{-1}$ & $50$ \\ 
$m$ & Vehicle mass & kg & 1300 \\ 
$I_z$ & Vertical moment of inertia  & $\textnormal{kg}\,\textnormal{m}^{2}$ & 2000 \\
$l_1$ & Front axle length & m & 1.4  \\
$l_2$ & Rear axle length & m & 1 \\
$F_{z1}$ & Front vertical force & N & $2.66 \cdot 10^3$ \\
$F_{z2}$ & Rear vertical force & N & $3.72 \cdot 10^3$ \\
$L_1$ & Front contact patch length & m & 0.11 \\
$L_2$ & Rear contact patch length & m & 0.09 \\
$\sigma_{0,1}$ & Front micro-stiffness & $\textnormal{m}^{-1}$ & 240 \\
$\sigma_{0,2}$ & Rear micro-stiffness & $\textnormal{m}^{-1}$ & 269 \\
$\phi_{1}$ & Front structural parameter & - & 0.92 \\
$\phi_{2}$ & Rear structural parameter & - & 0.92 \\
$\bar{\mu}_1(\cdot)$ & Front friction coefficient & -& 1\\
$\bar{\mu}_2(\cdot)$ & Rear friction coefficient & -& 1\\
$\bar{g}_1(\cdot)$ & Front friction function & -& 1\\
$\bar{g}_2(\cdot)$ & Rear friction function & -& 1\\
$\chi$ & Rear steering actuation & -& 0\\
$\varepsilon$ & Regularization parameter & - & 0 \\
\hline
\end{tabularx}
}
\caption{Model parameters}
\label{tab:parameters2}
\end{table}

Figure~\ref{figure:StatesOpen} illustrates the unstable behavior of the uncontrolled vehicle traveling at $v_x = 50$ $\textnormal{m}\,\textnormal{s}^{-1}$.
\begin{figure}
\centering
\includegraphics[width=1\linewidth]{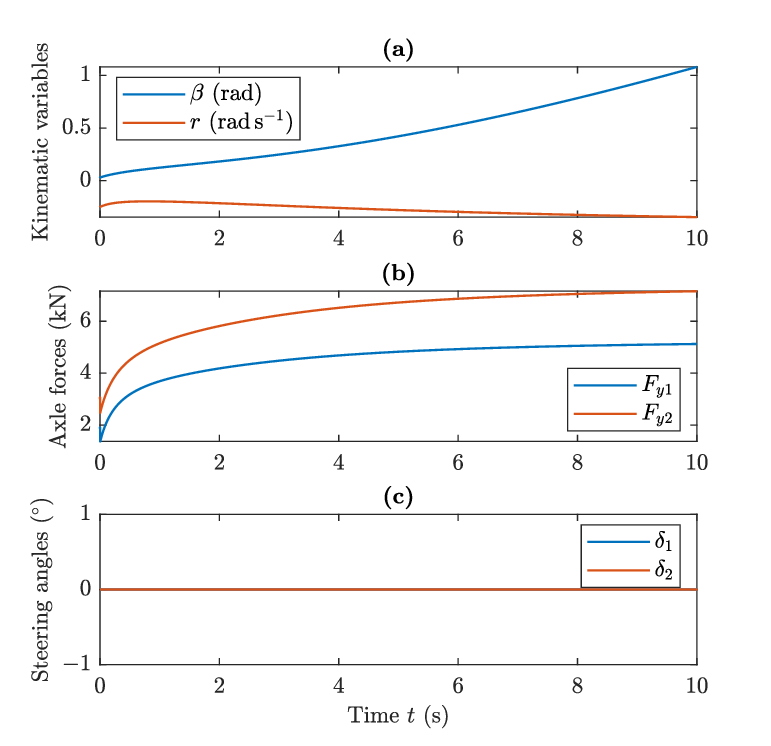} 
\caption{Open loop behavior of the lumped states and steering inputs: \textbf{(a)} kinematic variables; \textbf{(b)} axle forces; \textbf{(c)} steering inputs.}
\label{figure:StatesOpen}
\end{figure}
With the parameter values in Table~\ref{tab:parameters2}, the characteristic time scale of the reduced subsystem, computed as in~\eqref{eq:tau_xOpen}, is $\tau_X(0) = 2.76$ s at $\alpha^\star = 0$, whereas $\tau_z\approx 0.002$ s at $v_x=50$ $\text{m}\,\text{s}^{-1}$, which yields $\tau_z/\tau_X(0)\approx 7.24\times10^{-4}$, indicating a clear separation of time scales. Therefore, the vehicle was stabilized around $\alpha^\star = 0$, by setting 
\begin{align}\label{Fnum}
F(0) = \begin{bmatrix} 2.038 & -0.458\, \textnormal{s} \\ 0 & 0\end{bmatrix},
\end{align} 
which places the poles of $\tilde{A}_1^*(0)$ at $ -10$ and $ -15$ $\textnormal{s}^{-1}$, yielding $\tau_z/\tau_X^*(0) \approx 0.02$ and thus preserving the time scale separation in closed loop (note that the matrix $F(0)$ in~\eqref{Fnum} contains a zero row, reflecting the absence of rear steering). 
The closed-loop trends of the kinematic variables, axle forces, and steering inputs are depicted in Figure~\ref{figure:StatesClosed}\textbf{(a)}, \textbf{(b)}, and \textbf{(c)}, where a rapid convergence around zero may be observed for all the involved quantities. In particular, the steering angle $\delta_1(t)$ initially exhibits a pronounced transient, but never exceeds $4^\circ$ in absolute value, confirming \emph{a posteriori} the feasibility of the maneuver. 

Similar considerations may be drawn by inspecting Figure~\ref{figure:PDEs}, where closed-loop dynamics of the PDE states $z_\delta(\xi,t)$ are illustrated. In particular, under the proposed control action, the distributed states quickly approach zero, which is consistent with the observations reported above. Additional simulations were conducted by slightly varying the controller gain, without appreciable differences in the qualitative behavior of the closed-loop system. 
\begin{figure}
\centering
\includegraphics[width=1\linewidth]{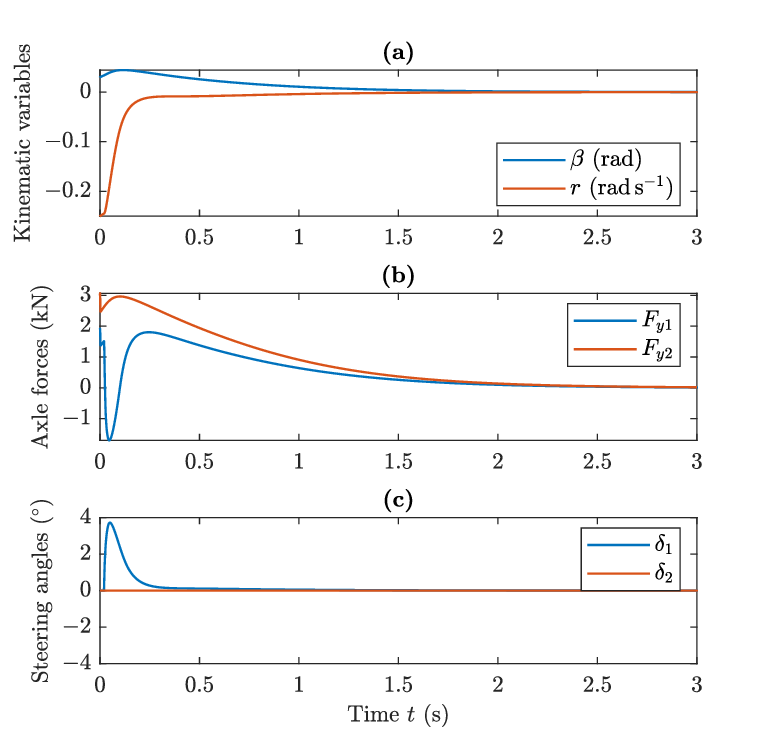} 
\caption{Closed loop behavior of the lumped states and steering inputs, for $\alpha^\star = 0$ and $F(0)$ as in~\eqref{Fnum}: \textbf{(a)} kinematic variables; \textbf{(b)} axle forces; \textbf{(c)} steering inputs.}
\label{figure:StatesClosed}
\end{figure}
\begin{figure*}
\centering
\includegraphics[width=0.7\linewidth]{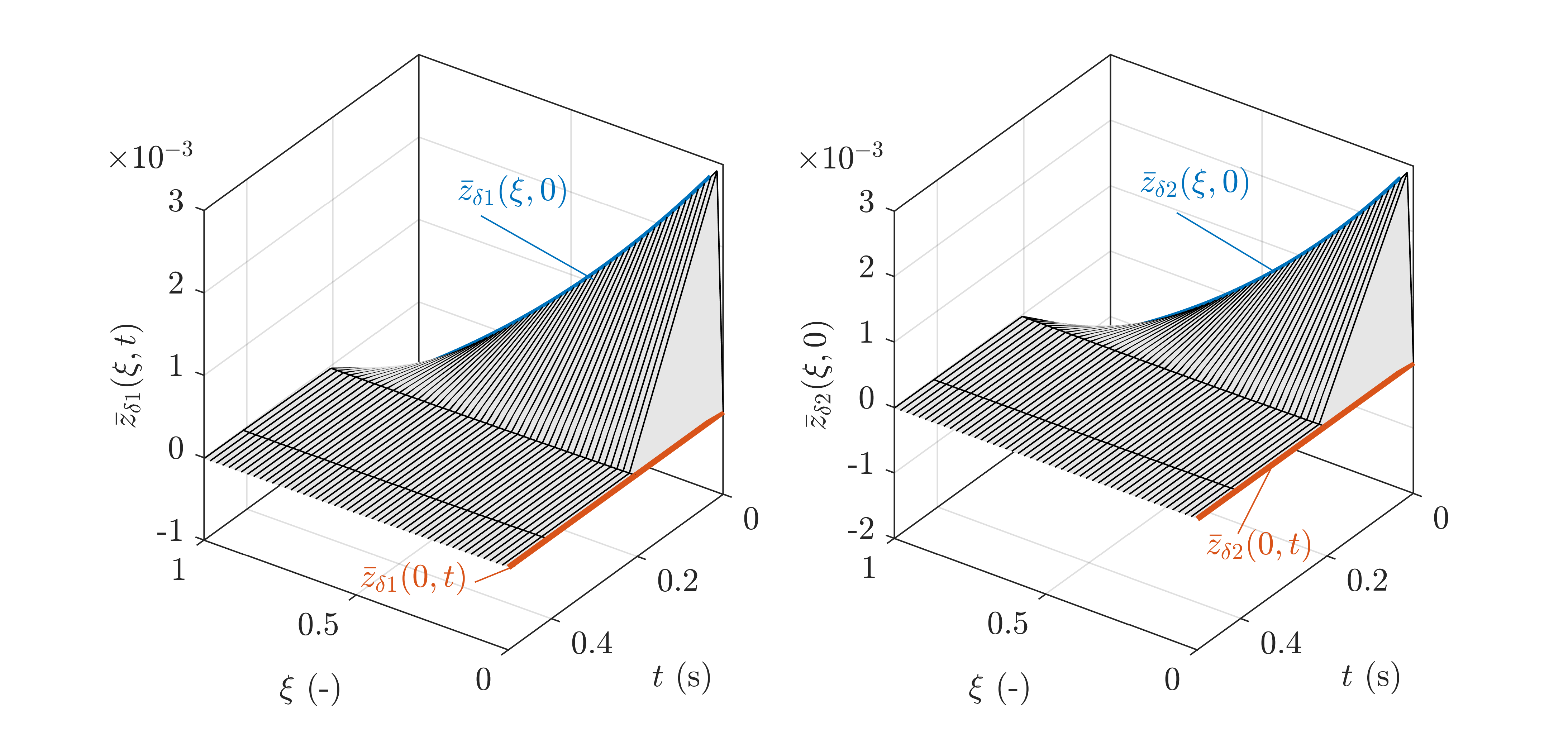} 
\caption{Evolution of the PDE states $z_\delta(\xi,t)$, along with its IC (blue lines) and BC (orange lines).}
\label{figure:PDEs}
\end{figure*}

The results asserted by Theorem~\ref{thm:main2} are of a local nature: ICs deviating largely from the target equilibrium may destabilize the closed-loop system's response. Additionally, relatively large values of the parameter $\tau_z$ -- which is essentially determined by the longitudinal speed $v_x$ for constant $L_1$ and $L_2$ -- may invalidate the conclusion of the singular perturbation analysis. In this context, the influence of different ICs $X_0 = -k[0.03\; -0.25\, \textnormal{rad}\,\textnormal{s}^{-1}]^{\mathrm{T}}$ on the closed-loop dynamics of the ODE-PDE system is illustrated in Figure~\ref{figure:ICs} for $k=1$, 3, and 6. As intuitively expected, ICs that are further from the target equilibrium may slow down the convergence of $\norm{(X(t)-X^\star, z_\delta(\cdot,t))}_{\altmathcal{X}}$, which is particularly noticeable for $k=6$. Finally, the effect of different values of $v_x$ is investigated graphically in Figures~\ref{figure:vx} and~\ref{figure:vx5}. In particular, considering $v_x = 10$, 20, and 50 $\textnormal{m}\,\textnormal{s}^{-1}$, Figure~\ref{figure:vx} demonstrates that the synthesized controller can successfully stabilize the ODE-PDE interconnection~\eqref{eq:originalSystems} on a wide range of longitudinal speeds.
A markedly different situation is illustrated in Figure~\ref{figure:vx5} for $v_x=1$ $\textnormal{m}\,\textnormal{s}^{-1}$. At this velocity, the open-loop ODE-PDE interconnection is asymptotically stable, although the time-scale indicator satisfies $\tau_z/\tau_X(0)\approx9.2$, showing that the tire and reduced vehicle dynamics are not separated. A feedback gain was subsequently selected to place the poles of the reduced closed-loop subsystem at $-5$ and $-7.5$ $\textnormal{s}^{-1}$, for which $\tau_z/\tau_X^*(0)\approx0.5$. Although the assigned reduced-order poles are real and strictly negative, the full ODE-PDE interconnection exhibits a more oscillatory transient than the open-loop system. This demonstrates that, outside a regime of pronounced time-scale separation, the reduced-order poles alone do not determine the transient response of the full interconnection. The feedback gain also modifies the coupling with the distributed tire dynamics through $A_2+G_2F(0)$. The results in Figure~\ref{figure:vx5} were obtained without input delay.

\begin{figure}
\centering
\includegraphics[width=0.9\linewidth]{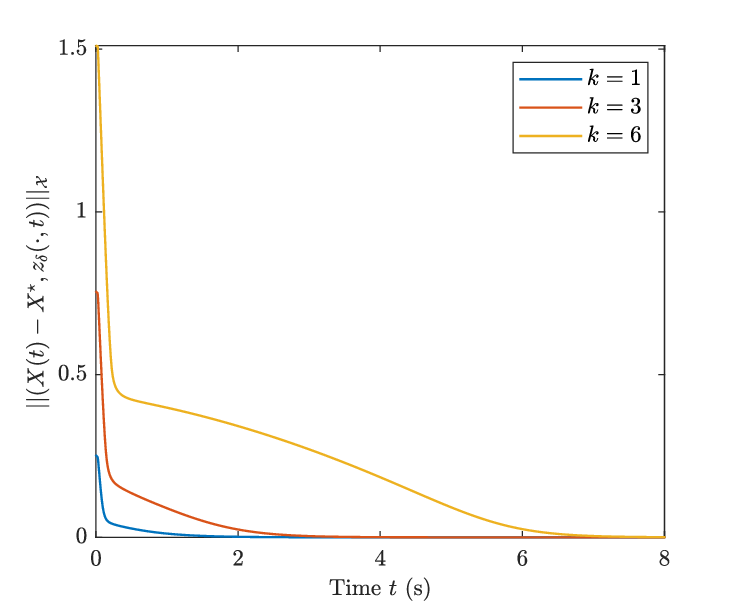} 
\caption{Closed loop behavior of $\norm{(X(t)-X^\star,z_\delta(\cdot,t))}_{\altmathcal{X}}$ for different values of $k = 1$, 3, and 6.}
\label{figure:ICs}
\end{figure}

\begin{figure}
\centering
\includegraphics[width=0.9\linewidth]{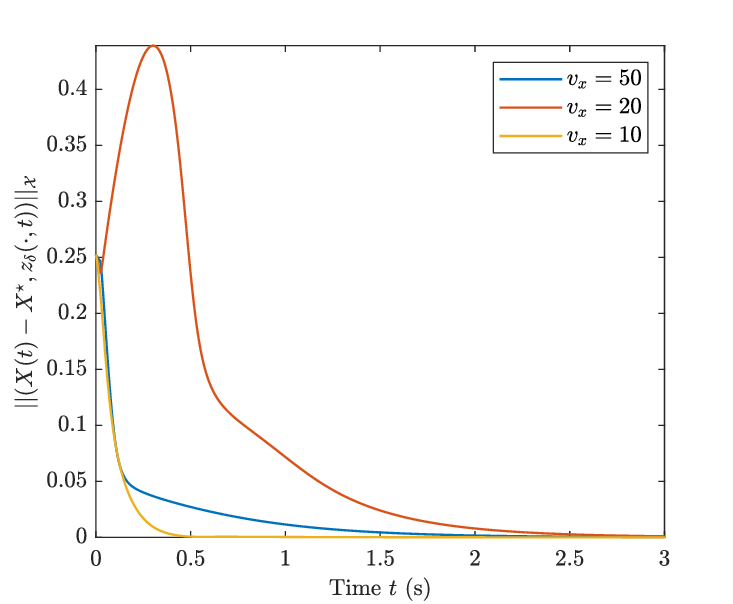} 
\caption{Closed loop behavior of $\norm{(X(t)-X^\star,z_\delta(\cdot,t))}_{\altmathcal{X}}$ for different values of $v_x = 10$, 20, and 50 $\textnormal{m}\,\textnormal{s}^{-1}$.}
\label{figure:vx}
\end{figure}

\begin{figure}
\centering
\includegraphics[width=0.9\linewidth]{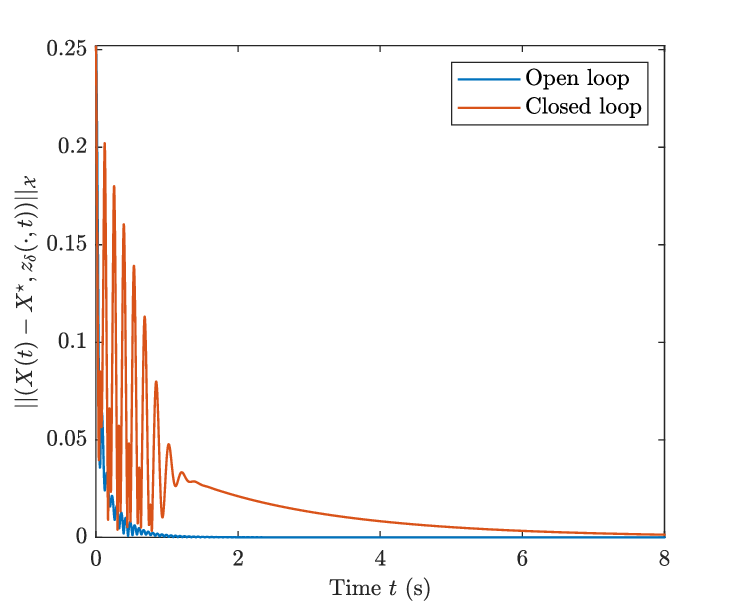} 
\caption{Open and closed loop behavior of $\norm{(X(t)-X^\star,z_\delta(\cdot,t))}_{\altmathcal{X}}$ for $v_x = 1$ $\textnormal{m}\,\textnormal{s}^{-1}$.}
\label{figure:vx5}
\end{figure}

\section{Conclusions}\label{sect:conclusion}

The present paper investigated the stability and stabilization of semilinear single-track vehicle models with distributed tire friction dynamics using singular perturbation analysis. Specifically, it was shown that when an appropriate time-scale parameter -- defined as the ratio between a characteristic length scale of the rolling contact problem and the vehicle's longitudinal speed -- is sufficiently small, standard finite-dimensional methods can successfully be applied to study local stability and stabilizability. Whilst the paper's results do not introduce new controller design strategies, they provide the first rigorous mathematical justification for long-standing practices in automotive research, establishing fundamental advances in vehicle dynamics and opening new perspectives for the control of semilinear ODE-PDE systems with distributed friction.

Future research will explore the design of nonlinear controllers, possibly in conjunction with adaptive observers for vehicle state and friction estimation, along with their impact on the dynamics of the full ODE-PDE interconnection. Additionally, the behavior of more complex vehicle models, incorporating the effect of large steering angles and suspension dynamics, should be studied.

\subsubsection*{Acknowledgments}
This research was financially supported by the project FASTEST (Reg. no. 2023-06511), funded by the Swedish Research Council.

\subsubsection*{Declaration of interest}
Declaration of interest: none.

{\setstretch{1} 
\printbibliography

@book{Pacejka2,
  author    = {Pacejka, H. B.},
  title     = {Tire and Vehicle Dynamics},
  edition   = {3},
  publisher = {Elsevier Butterworth-Heinemann},
  address   = {Amsterdam},
  year      = {2012}
}

@book{Guiggiani,
  author    = {Guiggiani, M.},
  title     = {The Science of Vehicle Dynamics},
  edition   = {3},
  publisher = {Springer},
  address   = {Cham},
  year      = {2023},
  doi       = {10.1007/978-3-031-06461-6}
}

@book{LibroMio,
  author    = {Romano, L.},
  title     = {Advanced Brush Tyre Modelling},
  series    = {SpringerBriefs in Applied Sciences},
  publisher = {Springer},
  address   = {Cham},
  year      = {2022},
  doi       = {10.1007/978-3-030-98435-9}
}

@article{Mojtaba1,
  author  = {Sharifzadeh, M. and Akbari, A. and Timpone, F. and Daryani, R.},
  title   = {Vehicle tyre/road interaction modeling and identification of its parameters using real-time trust-region methods},
  journal = {IFAC-PapersOnLine},
  volume  = {49},
  number  = {3},
  pages   = {111--116},
  year    = {2016},
  doi     = {10.1016/j.ifacol.2016.07.019}
}

@article{Mojtaba3,
  author  = {Sharifzadeh, M. and Senatore, A. and Farnam, A. and Akbari, A. and Timpone, F.},
  title   = {A real-time approach to robust identification of tyre--road friction characteristics on mixed-$\mu$ roads},
  journal = {Vehicle System Dynamics},
  volume  = {57},
  number  = {9},
  pages   = {1338--1360},
  year    = {2019},
  doi     = {10.1080/00423114.2018.1504974}
}

@inproceedings{Shao1,
  author    = {Shao, L. and Jin, C. and Lex, C. and Eichberger, A.},
  title     = {Nonlinear adaptive observer for side slip angle and road friction estimation},
  booktitle = {Proceedings of the IEEE Conference on Decision and Control},
  pages     = {6258--6265},
  year      = {2016}
}

@article{Shao3,
  author  = {Shao, L. and Jin, C. and Lex, C. and Eichberger, A.},
  title   = {Robust road friction estimation during vehicle steering},
  journal = {Vehicle System Dynamics},
  volume  = {57},
  number  = {4},
  pages   = {493--519},
  year    = {2019}
}

@article{Basilio2,
  author  = {Guo, N. and Zhang, X. and Zou, Y. and Lenzo, B.},
  title   = {A supervisory control strategy of distributed drive electric vehicles for coordinating handling, lateral stability, and energy efficiency},
  journal = {IEEE Transactions on Transportation Electrification},
  volume  = {7},
  number  = {4},
  pages   = {2488--2504},
  year    = {2021}
}

@article{Basilio4,
  author  = {Lenzo, B. and Zanchetta, M. and Sorniotti, A.},
  title   = {Yaw Rate and Sideslip Angle Control Through Single Input Single Output Direct Yaw Moment Control},
  journal = {IEEE Transactions on Control Systems Technology},
  year    = {2020}
}

@article{LateralControl,
  author  = {Li, Z. and Chen, H. and Liu, H. and Wang, P. and Gong, X.},
  title   = {Integrated Longitudinal and Lateral Vehicle Stability Control for Extreme Conditions With Safety Dynamic Requirements Analysis},
  journal = {IEEE Transactions on Intelligent Transportation Systems},
  volume  = {23},
  number  = {10},
  pages   = {19285--19298},
  year    = {2022},
  doi     = {10.1109/TITS.2022.3152485}
}

@article{Savaresi,
  author  = {Tanelli, M. and Astolfi, A. and Savaresi, S. M.},
  title   = {Robust nonlinear output feedback control for brake by wire control systems},
  journal = {Automatica},
  volume  = {44},
  number  = {4},
  pages   = {1078--1087},
  year    = {2008},
  doi     = {10.1016/j.automatica.2007.08.020}
}

@article{FEM,
  author  = {Romano, L. and Målqvist, A.},
  title   = {Finite element modelling of linear rolling contact problems},
  journal = {ESAIM: M2AN},
  volume  = {58},
  number  = {5},
  pages   = {	1541--1579},
  year    = {2024},
  doi     = {10.1051/m2an/2024052}
}

@article{Gerdes3,
  author  = {Talbot, J. and Brown, M. and Gerdes, J. C.},
  title   = {Shared Control Up to the Limits of Vehicle Handling},
  journal = {IEEE Transactions on Intelligent Vehicles},
  volume  = {9},
  number  = {1},
  pages   = {2977--2987},
  year    = {2023},
  doi     = {10.1109/TIV.2023.3300989}
}

@article{IEEEVT1,
  author  = {Du, H. and Zhang, N. and Dong, G.},
  title   = {Stabilizing Vehicle Lateral Dynamics With Considerations of Parameter Uncertainties and Control Saturation Through Robust Yaw Control},
  journal = {IEEE Transactions on Vehicular Technology},
  volume  = {59},
  number  = {5},
  pages   = {2593--2597},
  year    = {2010},
  doi     = {10.1109/TVT.2010.2045520}
}

@article{IEEEVT2,
  author  = {Zhou, H. and Jia, F. and Jing, H. and Liu, Z. and G{\"u}ven{\c{c}}, L.},
  title   = {Coordinated Longitudinal and Lateral Motion Control for Four Wheel Independent Motor-Drive Electric Vehicle},
  journal = {IEEE Transactions on Vehicular Technology},
  volume  = {67},
  number  = {5},
  pages   = {3782--3790},
  year    = {2018},
  doi     = {10.1109/TVT.2018.2816936}
}

@article{IEEEVT3,
  author  = {Nam, K. and Fujimoto, H. and Hori, Y.},
  title   = {Lateral Stability Control of In-Wheel-Motor-Driven Electric Vehicles Based on Sideslip Angle Estimation Using Lateral Tire Force Sensors},
  journal = {IEEE Transactions on Vehicular Technology},
  volume  = {61},
  number  = {5},
  pages   = {1972--1985},
  year    = {2012},
  doi     = {10.1109/TVT.2012.2191627}
}

@article{Takacs5,
  author  = {Tak{\'a}cs, D. and St{\'e}p{\'a}n, G.},
  title   = {Contact patch memory of tyres leading to lateral vibrations of four-wheeled vehicles},
  journal = {Philosophical Transactions of the Royal Society A},
  volume  = {371},
  pages   = {20120427},
  year    = {2013},
  doi     = {10.1098/rsta.2012.0427}
}

@article{Takacs2,
  author  = {Tak{\'a}cs, D. and St{\'e}p{\'a}n, G.},
  title   = {Micro-shimmy of towed structures in experimentally uncharted unstable parameter domain},
  journal = {Vehicle System Dynamics},
  volume  = {50},
  number  = {11},
  pages   = {1613--1630},
  year    = {2012}
}

@article{Takacs1,
  author  = {Tak{\'a}cs, D. and Orosz, G. and St{\'e}p{\'a}n, G.},
  title   = {Delay effects in shimmy dynamics of wheels with stretched string-like tyres},
  journal = {European Journal of Mechanics A/Solids},
  volume  = {28},
  number  = {3},
  pages   = {516--525},
  year    = {2009}
}

@article{Takacs3,
  author  = {Tak{\'a}cs, D. and St{\'e}p{\'a}n, G. and Hogan, S. J.},
  title   = {Isolated large amplitude periodic motions of towed rigid wheels},
  journal = {Nonlinear Dynamics},
  volume  = {52},
  pages   = {27--34},
  year    = {2008},
  doi     = {10.1007/s11071-007-9253-y}
}

@book{Zwart,
  author    = {Curtain, R. and Zwart, H.},
  title     = {Introduction to Infinite-Dimensional Systems Theory: A State-Space Approach},
  edition   = {1},
  publisher = {Springer},
  address   = {New York},
  year      = {2020}
}

@article{Weiss,
  author  = {Weiss, G. and Curtain, R.},
  title   = {Dynamic stabilization of regular linear systems},
  journal = {IEEE Trans. Autom. Control},
  volume  = {42},
  number  = {1},
  pages   = {4--21},
  year    = {1997},
  doi     = {10.1109/9.553684}
}

@book{Coron,
  author    = {Bastin, G. and Coron, J.-M.},
  title     = {Stability and Boundary Stabilization of 1-D Hyperbolic Systems},
  publisher = {Birkh{\"a}user},
  address   = {Cham},
  year      = {2016},
  doi       = {10.1007/978-3-319-32062-5}
}

@book{KrsticBook1,
  author    = {Krsti{\'c}, M. and Smyshlyaev, A.},
  title     = {Boundary Control of PDEs: A Course on Backstepping Designs},
  publisher = {SIAM},
  year      = {2008}
}

@book{KrsticBook2,
  author    = {Vazquez, R. and Krsti{\'c}, M.},
  title     = {Control of Turbulent and Magnetohydrodynamic Channel Flows},
  publisher = {Birkh{\"a}user},
  address   = {Boston},
  year      = {2007}
}

@book{OleBook,
  author    = {Anfinsen, H. and Aamo, O. M.},
  title     = {Adaptive Control of Hyperbolic PDEs},
  publisher = {Springer},
  address   = {Cham},
  year      = {2019}
}

@article{Vazquez,
  author  = {Vazquez, R. and Krsti{\'c}, M.},
  title   = {Explicit integral operator feedback for local stabilization of nonlinear thermal convection loop PDEs},
  journal = {Syst. Control Lett.},
  volume  = {55},
  number  = {8},
  pages   = {624--632},
  year    = {2006},
  doi     = {10.1016/j.sysconle.2005.09.019}
}

@article{Tikhonov1,
  author  = {Tang, Y. and Prieur, C. and Girard, A.},
  title   = {Tikhonov theorem for linear hyperbolic systems},
  journal = {Automatica},
  volume  = {57},
  pages   = {1--10},
  year    = {2015},
  doi     = {10.1016/j.automatica.2015.03.028}
}

@article{Tikhonov2,
  author  = {Tang, Y. and Mazanti, G.},
  title   = {Stability analysis of coupled linear ODE-hyperbolic PDE systems with two time scales},
  journal = {Automatica},
  volume  = {85},
  pages   = {386--396},
  year    = {2017},
  doi     = {10.1016/j.automatica.2017.07.052}
}

@inproceedings{Tikhonov3,
  author    = {Cerpa, E. and Prieur, C.},
  title     = {Effect of time scales on stability of coupled systems involving the wave equation},
  booktitle = {Proc. IEEE Conf. Decis. Control},
  pages     = {1236--1241},
  year      = {2017},
  doi       = {10.1109/CDC.2017.8263825}
}

@article{DistrLuGre,
  author  = {Romano, L. and Aamo, O. M. and {\AA}slund, J. and Frisk, E.},
  title   = {Stability and dissipativity of the distributed LuGre friction model},
  journal = {IEEE Control Syst. Lett.},
  year    = {2025},
  doi     = {10.1109/LCSYS.2025.3572419}
}

@unpublished{FrBD,
  author  = {Romano, L. and Aamo, O. M. and {\AA}slund, J. and Frisk, E.},
  title   = {First-order friction models with bristle dynamics: lumped and distributed formulations},
journal = {arXiv preprint},
  note    = {Submitted to IEEE Trans. Control Syst. Technol.},
  year    = {2026}
}

@article{Rill,
  author  = {Rill, G. and Schaeffer, T. and Schuderer, M.},
  title   = {LuGre or not LuGre},
  journal = {Multibody Syst. Dyn.},
  volume  = {60},
  pages   = {191--218},
  year    = {2024},
  doi     = {10.1007/s11044-023-09909-5}
}

@article{Rill0,
  author  = {Schuderer, M. and Rill, G. and Schaeffer, T.},
  title   = {Friction modeling from a practical point of view},
  journal = {Multibody Syst. Dyn.},
  volume  = {63},
  pages   = {141--158},
  year    = {2025},
  doi     = {10.1007/s11044-024-09978-0}
}

@article{BicyclePDE,
  author  = {Romano, L. and Aamo, O. M. and {\AA}slund, J. and Frisk, E.},
  title   = {Stability analysis of linear single-track models with transient tyre dynamics},
  journal = {Veh. Syst. Dyn.},
  pages   = {1--35},
  year    = {2024},
  doi     = {10.1080/00423114.2024.2445163}
}

@article{SemilinearV,
  author  = {Romano, L. and Aamo, O. M. and {Åslund}, J. and Frisk, E.},
  title   = {Semilinear single-track models with distributed tyre friction dynamics},
  journal = {Nonlinear Dyn.},
  year    = {2026},
  volume  = {114},
number = {138}
}

@article{PassExp,
  author = {Romano, L. and Aamo, O. M. and Krsti{\'c}, M. and {\AA}slund, J. and Frisk, E.},
  title  = {Passivity-exploiting stabilization of semilinear single-track vehicle models with distributed tire friction dynamics},
  journal  = {Automatica},
 volume = {193},
number = {113223},
  year   = {2026}
}

@book{Pazy,
  author    = {Pazy, A.},
  title     = {Semigroups of Linear Operators and Applications to Partial Differential Equations},
  publisher = {Springer},
  address   = {New York},
  year      = {1983},
  doi       = {10.1007/978-1-4612-5561-1}
}

@inproceedings{TsiotrasConf,
  author    = {Canudas-de-Wit, C. and Tsiotras, P.},
  title     = {Dynamic tire friction models for vehicle traction control},
  booktitle = {Proc. IEEE Conf. Decis. Control},
  pages     = {3746--3751},
  year      = {1999},
  doi       = {10.1109/CDC.1999.827937}
}

@article{Tsiotras1,
  author  = {Canudas-de-Wit, C. and Tsiotras, P. and Velenis, E.},
  title   = {Dynamic Friction Models for Road/Tire Longitudinal Interaction},
  journal = {Veh. Syst. Dyn.},
  volume  = {39},
  number  = {3},
  pages   = {189--226},
  year    = {2003}
}

@article{Tsiotras2,
  author  = {Velenis, E. and Tsiotras, P. and Canudas-de-Wit, C. and Sorine, M.},
  title   = {Dynamic tyre friction models for combined longitudinal and lateral vehicle motion},
  journal = {Veh. Syst. Dyn.},
  volume  = {43},
  number  = {1},
  pages   = {3--29},
  year    = {2005}
}

@article{Deur0,
  author  = {Deur, J.},
  title   = {Modeling and Analysis of Longitudinal Tire Dynamics Based on the Lugre Friction Model},
  journal = {IFAC Proc. Vol.},
  volume  = {34},
  number  = {1},
  pages   = {91--96},
  year    = {2001},
  doi     = {10.1016/S1474-6670(17)34383-5}
}

@article{Deur1,
  author  = {Deur, J. and Asgari, J. and Hrovat, D.},
  title   = {A 3D Brush-type Dynamic Tire Friction Model},
  journal = {Veh. Syst. Dyn.},
  volume  = {42},
  number  = {3},
  pages   = {133--173},
  year    = {2004}
}

@article{Deur2,
  author  = {Deur, J. and Ivanovi{\'c}, V. and Troulis, M.},
  title   = {Extensions of the LuGre tyre friction model related to variable slip speed along the contact patch length},
  journal = {Veh. Syst. Dyn.},
  volume  = {43},
  pages   = {508--524},
  year    = {2005}
}

@article{Beregi1,
  author  = {Beregi, S. and Tak{\'a}cs, D. and St{\'e}p{\'a}n, G.},
  title   = {Tyre induced vibrations of the car--trailer system},
  journal = {J. Sound Vib.},
  volume  = {362},
  pages   = {215--227},
  year    = {2016}
}

@article{Beregi3,
  author  = {Beregi, S. and Tak{\'a}cs, D. and Gyebroszki, G.},
  title   = {Theoretical and experimental study on the nonlinear dynamics of wheel-shimmy},
  journal = {Nonlinear Dyn.},
  volume  = {98},
  pages   = {2581--2593},
  year    = {2019},
  doi     = {10.1007/s11071-019-05225-w}
}
}

\end{document}